%% file: main.tex
\documentclass[11pt]{article}
\usepackage[utf8]{inputenc}

\input{headers}
\allowdisplaybreaks

\title{Automating Food Drop: The Power of Two Choices for Dynamic and Fair Food Allocation}
\date{}

\author{%
  Marios Mertzanidis \\
  Purdue University\\
  \texttt{mmertzan@purdue.edu}\\ 
  \and
  Alexandros Psomas \\
  Purdue University\\
  \texttt{apsomas@cs.purdue.edu}\\
  \and
  Paritosh Verma \\
  Purdue University\\
  \texttt{verma136@purdue.edu}
}

\begin{document}

\maketitle

\begin{abstract}
Food waste and food insecurity are two closely related pressing global issues. Food rescue organizations worldwide run programs aimed at addressing the two problems. In this paper, we partner with a non-profit organization in the state of Indiana that leads \emph{Food Drop}, a program that is designed to redirect rejected truckloads of food away from landfills and into food banks. The truckload to food bank matching decisions are currently made by an employee of our partner organization. In addition to this being a very time-consuming task, as perhaps expected from human-based matching decisions, the allocations are often skewed: a small percentage of the possible recipients receives the majority of donations. Our goal in this partnership is to completely automate Food Drop.
In doing so, we need a matching algorithm for making real-time decisions that strikes a balance between ensuring fairness for the food banks that receive the food and optimizing efficiency for the truck drivers.
In this paper, we describe the theoretical guarantees and experiments that dictated our choice of algorithm in the platform we built and deployed for our partner organization.
Our work also makes contributions to the literature on load balancing and balls-into-bins games, that might be of independent interest. Specifically, we study the allocation of $m$ weighted balls into $n$ weighted bins, where each ball has two non-uniformly sampled random bin choices, and prove upper bounds, that hold with high probability, on the maximum load of any bin.
\end{abstract}

\input{intro}

\input{prelims}

\input{theory}

\input{Majorization}

\input{WeightedBalls}

\input{experiments}

\section*{Acknowledgements}

Marios Mertzanidis, Alexandros Psomas, and Paritosh Verma are supported in part by an NSF CAREER award CCF-2144208, a Google AI for Social Good award, and research awards from Google and Supra.

\bibliographystyle{alpha}
\bibliography{refs.bib}

\appendix

\input{appendix}

\end{document}

%% file: headers.tex
\usepackage{tikz}
\usepackage{tikz-dependency}
\usepackage{thm-restate}

\usepackage{bbm}

\usepackage[numbers]{natbib}
\usepackage{xspace}
\usepackage[margin=1in]{geometry}
\usepackage{pstricks}
\usepackage{comment}
\usepackage{dsfont}
\usepackage{tikz}
\usepackage{tikz-dependency}

\usepackage{amsmath,amssymb,amsthm,amsfonts,latexsym,bbm,xspace,graphicx,float,mathtools,mathdots}
\usepackage{braket,caption,subcaption,ellipsis,xcolor,textcomp}
\usepackage[colorlinks,citecolor=blue]{hyperref}
\usepackage{enumitem}
\usepackage[ruled]{algorithm2e}
\usepackage{algpseudocode}
\usepackage{listings}

\usepackage{thmtools}
\usepackage{thm-restate}
\usepackage[capitalise]{cleveref}
\usepackage[colorinlistoftodos,disable]{todonotes}

\newtheorem{theorem}{Theorem}
\declaretheorem[name=Theorem,sibling=theorem]{thm-restate}
\newtheorem{lemma}{Lemma}
\newtheorem{claim}{Claim}

\newtheorem{corollary}{Corollary}

\theoremstyle{definition}
\newtheorem{definition}{Definition}

\usepackage{multirow} 

\newcommand{\Prob}[2]{\mathrm{Pr}_{#1}\left[#2\right]}

\newcommand{\EX}[2]{\underset{#1}{\mathbb{E}}\left[ #2 \right]}

\newcommand{\R}{\mathbbm R}

\newcommand{\calD}{\mathcal{D}}

\newcommand{\calF}{\mathcal{F}}

\newcommand{\calM}{\mathcal{M}}

\newcommand{\calW}{\mathcal{W}}

\DeclareMathOperator*{\argmin}{arg\,min}



%% file: intro.tex
\section{Introduction}


Food waste, the act of discarding safe, high-quality food, instead of eating it, is a massive, worldwide issue. Overall, about a third of the world's food is thrown away~\cite{gustavsson2011global}.
In the United States, between 30 and 40 percent of food is wasted~\cite{USDA}, which roughly equals 149 billion meals~\cite{feedingAmericaWaste}. At the same time, almost paradoxically, food insecurity, defined by the United States Department of Agriculture as ``limited or uncertain availability of nutritionally adequate and safe foods or limited or uncertain ability to acquire acceptable foods in socially acceptable ways''~\cite{bickel2000guide}, persists as a pressing global challenge. In 2021, close to 12 percent of the global population, or, equivalently, 928 million people, were food insecure~\cite{unicef2021state}. In the United States, 44.2 million people lived in food-insecure households in 2022~\cite{usdaSecurity2023}, affecting every single state, and, in fact, every single county in the country~\cite{feedingAmericaHunger}.

Food rescue organizations worldwide are leading programs aimed at addressing food waste and food insecurity. \emph{Food Drop} is such a program, run by the non-governmental organization \emph{Indy Hunger Network} (IHN), in the state of Indiana. Food Drop matches truck drivers with rejected truckloads of food---food that would otherwise be discarded at a landfill---to food banks. The average amount of food matched per month is 10,447 lbs, with some months having as many as 62,000 lbs of food. Matching decisions are currently made by the Food Assistance Programs Manager of IHN, who also needs to manually check availability and willingness to accept each donation from the food bank's side, facilitate the exchange of contact information between the food bank and truck driver, and so on. In addition to this being a very time-consuming task, similarly to other allocation problems in the food space, e.g.,~\cite{lee2019webuildai}, the allocations of Food Drop are often skewed: a small percentage of the possible recipients receives the majority of donations. In this work, we partner with IHN with the goal of completely automating Food Drop.

The main focus of this paper is to present the model, theoretical results, and experiments that dictated our choice of algorithm for matching drivers to food banks for the platform we built for IHN. Our platform was deployed in April 2024~\cite{FoodDrop}. Before we describe these results, we briefly describe the platform for automated Food Drop.

\paragraph{Our platform}

\begin{figure}[hbt]
  \centering
  \includegraphics[width=1\linewidth]{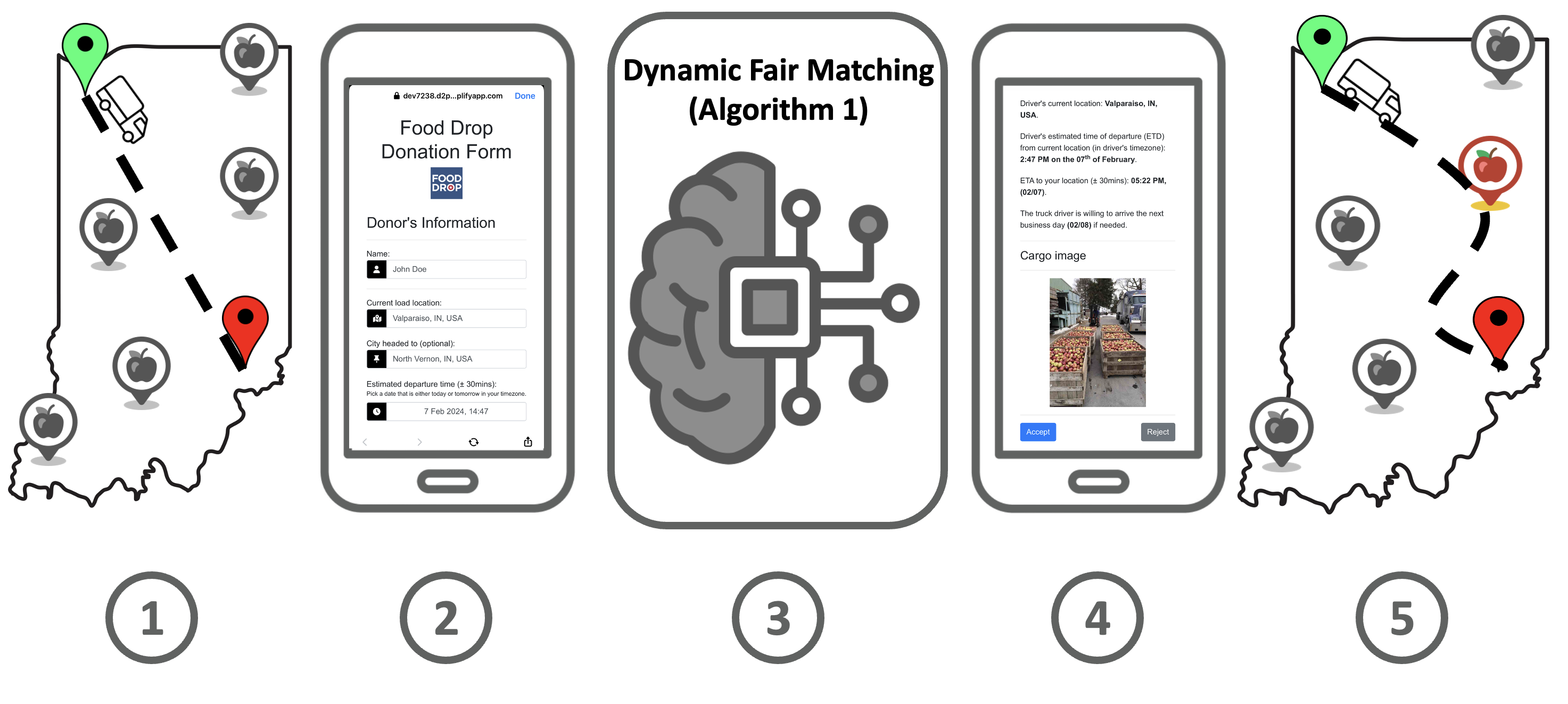} 
  \caption{A typical use case of our platform. 1) A truck driver with a  rejected food load wants to go from location A to location B. 2) They fill out a form in the online platform providing relevant details (see screenshot from our platform). 3) The back end runs~\Cref{ourAlgorithm} and matches the request to a food bank. 4) The representative of the matched food bank receives a notification on their phone, via SMS, through which they can access the information regarding the food load (see screenshot from our platform). 5) Once the food bank accepts the delivery, the truck driver is notified and given the contact information of the matched food bank, via SMS.}
  \label{fig:example}
\end{figure}

Our platform works as follows (see~\Cref{fig:example}). A driver that wants to make a donation fills out a form we designed, which elicits their information, including current location, destination, estimated departure time, weight and type of food they want to donate, as well as the reason the food load was rejected and a picture of the load. Our back end, currently hosted at Amazon AWS, runs an algorithm to match the request to a food bank. A representative from the food bank receives a notification on their phone, via SMS, which includes a link they can follow to access the information regarding the food load, as well as accept or reject. If they reject, we identify a new recipient, until one accepts.\footnote{In the current implementation, we default to a human after the fifth rejection.} Once a match is made, the truck driver receives an SMS notification that includes the contact information of their matched food bank. The code of our platform is open source, and can be found in a GitHub repository~\cite{IHNRepo}.
The one question that remains unanswered, the main question we address in this paper, is which algorithm should we use to match drivers with food banks.

\paragraph{A new problem in dynamic fair division}

In order to formally argue about our choice of algorithm, we develop the following, new model for dynamic fair division. There is an undirected graph of $n$ nodes, corresponding to counties, and each node $v$ has a population of food-insecure individuals $p_v$. Every edge in the graph has a weight, corresponding to the distance between counties; these distances satisfy the triangle inequality. A subset $\mathcal{F}$ of the nodes have food banks that can accept food donations, and every food-insecure individual is served by the food bank closest to them. Drivers appear over time, one in each step, for a total of $m$ steps, and need to be matched, immediately and irrevocably, to a food bank. Drivers have an origin node and a destination node, that are picked independently from a distribution $\calD$. We assume that $\calD$ selects node $v$ with probability \emph{approximately} $p_v / \sum_{u} p_u$, and specifically, probability in the range $[\frac{p_v}{\alpha \cdot \sum_{u} p_u} , \frac{\beta \cdot  p_v}{\sum_{u} p_u}]$. We call this an $(\alpha, \beta)$-biased distribution. Our fairness guarantees depend on $\alpha$ and $\beta$, which in practice we expect to be small. In our platform we expect drivers to have an origin and destination at large grocery stores and distribution centers; the number of these is proportional to the total population of a county, which, in turn, correlates strongly with the food-insecure population. The validity of this assumption is further strengthened by our data analysis presented in \cref{section:alphaBeta}.
Finally, drivers carry a donation of food, whose value is drawn from a distribution $\calW$.
When a driver delivers a donation to a food bank, the value of the donation is split equally among all the individuals the food bank serves (i.e., all the individuals for whom this is the nearest food bank).

As opposed to other problems in dynamic fair division, here we need to account for two-sided preferences: the preferences of the drivers and the preferences of the individuals served by food banks.
On one hand, we would like for drivers to not travel a lot. We say that a matching algorithm is $\delta$-driver efficient if the total distance a driver is asked to travel is within a $\delta$ factor of their optimal route, defined as the shortest path from the driver's origin to the driver's destination that goes through some food bank. On the other hand, we would like to provide fairness guarantees to the individuals served by the food banks. We say that a matching algorithm is $\varepsilon$-multiplicatively envy-free\footnote{We include ``multiplicatively'' to our definition because approximate envy-freeness typically has an additive, not a multiplicative, error.} if the total value received by an individual served by food bank $f$ at the end of the process is at least $1/\varepsilon$ times the value received by an individual served by a different food bank $f'$.

Our first result is a lower bound (\Cref{thm: lower bound}): even for $(1,1)$-biased distributions and donations with unit value, no matching algorithm is $(3-\delta)$-driver efficient, for $\delta>0$, and $\varepsilon$-multiplicatively envy-free, for $\varepsilon < n$. 
That is, non-trivial fairness guarantees are incompatible with $(3-\delta)$-driver efficiency.
Our main result (\Cref{thm: main result}) is a \emph{matching} upper bound.
We give a $3$-driver efficient algorithm that is $(1 + \varepsilon)$-multiplicatively envy-free with high probability, where, under some mild technical assumptions, $\varepsilon \in O \left( \frac{N \log N}{m - N \log N} \right)$ for unit value donations, and $\varepsilon \in O \left( \frac{N \log N}{m - N\sqrt{m}} \right)$ for donations whose values are drawn from a distribution $\calW$ such that $\EX{}{e^{\lambda \calW}} < \infty$, for some constant $\lambda$, and $\EX{}{\calW}=1$ (the latter assumption is without loss of generality).

Our algorithm itself,~\Cref{ourAlgorithm}, is quite simple to describe and implement, which is another very compelling reason for using it in our platform. It works as follows: for each driver, find the food bank $f_o$ that is closest to the driver's origin and the food bank $f_d$ that is closest to the driver's destination. Match the driver to the one of $f_o$ and $f_d$ whose served individuals have received less value so far.
The $3$ approximation guarantee of this algorithm follows from the triangle inequality. The fairness guarantee boils down to a new statement about the standard balls-into-bins problem. Since this result might be of independent interest, we describe it using balls-into-bins terminology. We start by giving a brief background on balls-into-bins processes.

\paragraph{Balls-into-bins: background}

Balls-into-bins processes represent a fundamental paradigm in probability theory, with a myriad of applications in computer science and operations research. The literature on these processes is extensive; here, to provide some necessary context, we briefly mention some of the results that are closely related to our work. For an in-depth survey, we refer the reader to~\cite{sitaraman2001power}.


Consider $n$ balls thrown into $n$ bins. It is well understood that assigning each ball to a uniformly random bin yields an allocation in which the heaviest bin won't have more than $(1+O(1)) \ln n / \ln \ln n $ balls, with high probability~\cite{gonnet1981expected}. In seminal work,~\cite{azar1994balanced} prove that choosing two bins uniformly at random and placing the ball in the least loaded bin yields an exponential improvement: the heaviest bin won't have more than $\ln \ln n/ \ln 2 +O(1)$ balls, with high probability.  
When generalizing to $m > n$ balls into $n$ bins, the so-called ``heavy loaded case,'' in a breakthrough result,~\cite{berenbrink2006balanced} prove that the maximum load is bounded by $m/n + O(\ln \ln n)$, with high probability. Amazingly, the deviation from $m/n$ does not depend on the number of balls inserted, a major difference with the random selection process which yields a $m/n + O(\sqrt{m \ln n/n})$ bound. A simpler and slightly weaker analysis was later given by~\cite{talwar2014balanced}. 

The aforementioned results assume that each bin is unit-sized and has an equal likelihood of being selected. Furthermore, balls are also of unit size. In many applications, these assumptions are quite restrictive.~\cite{byers2004geometric}, and later on,~\cite{wieder2007Heterogeneous}, study the problem under non-uniform selection probability of the bins.~\cite{Berenbrink2014non-uniform} consider non-uniform weights on bins. Regarding the weight of the inserted balls,~\cite{talwar2007weighted} and~\cite{yuval2015} study the problem where the balls' weights are sampled from a distribution. All these works prove similar bounds to the classic setting, illustrating the robustness of the paradigm.

\paragraph{Balls-into-bins: our contribution}
In this paper, we investigate the balls-into-bins process within the framework of non-uniformly weighted bins, non-uniform selection probabilities, and random weighting of each ball, drawn from an arbitrary (under mild assumptions) probability distribution.
Our work combines elements from all the aforementioned works. Specifically, we draw intuition from the potential function arguments of~\cite{yuval2015} (and \cite{talwar2014balanced}) and the coupling arguments of~\cite{wieder2007Heterogeneous} and \cite{Berenbrink2014non-uniform}. Unfortunately, these elements cannot be trivially combined to give results for our problem.~\cite{Berenbrink2014non-uniform}'s results on weighted bins heavily rely on a coupling and majorization argument, which only works when each bin is selected with a probability exactly proportional to its weight. On the other hand,~\cite{wieder2007Heterogeneous}'s coupling and majorization technique for handling non-uniform selection of bins cannot be used when the balls have different weights. Finally, the potential function arguments of~\cite{yuval2015}, that study weighted balls and the so-called $(1 + \beta)$-choice process, heavily rely on the fact that each bin is uniformly weighted. In this work, we introduce non-trivial steps that allow us to combine all three approaches. On a high level, we bound the change in every step for a potential function that depends on the load of each bin, borrowing from~\cite{yuval2015}. To handle weighted bins, we think of a bin of weight $N_i$ as being composed of $N_i$ slots. Instead of bounding the potential function for each bin, we bound it for each slot, using a  weaker, intermediate process that adds each ball to the selected slot but does not equally distribute the load among the slots of the corresponding bin. Finally, regarding the $(\alpha, \beta)$-biased sampling distribution, we prove upper and lower bounds on the probability we select the $i^{th}$ most loaded slot; we can delicately plug these bounds in our analysis without occurring a super-constant loss.

\paragraph{Our experiments}

In~\Cref{sec: experiments} we evaluate our algorithm, along with some natural benchmarks, across different states, using Feeding America's Map the Meal Gap 2023 dataset~\cite{MapMeal23} on food insecurity, and the locations of Feeding America's partner food banks. We see that our algorithm performs consistently well. For example, in the state of Indiana, the maximum multiplicative envy quickly converges to a number close to $1$, while on average drivers are only asked to drive 12\% more compared to their optimal route. We compare our algorithm's performance with the performance of a natural parameterized family of algorithms that interpolates between the optimal for fairness and optimal for driver efficiency algorithms.
We observe that, even though in some states in the US there are algorithms within this parameterized family that achieve comparable results to ours, in several other states, e.g. in the state of California, our algorithm achieves superior trade-offs between fairness and efficiency. That is, algorithms with comparable fairness guarantees are a lot less efficient, and algorithms with comparable efficiency guarantees are a lot less fair. This suggests that our approach offers robustness across diverse geographical contexts.

\subsection{Additional Related Work}

Beyond the aforementioned literature on the balls-into-bins process, our work contributes to the growing literature on dynamic or online fair division~\citep{kash2014no,aleksandrov2015online,friedman2015dynamic,friedman2017controlled,benade2018make,he2019achieving,zeng2020fairness,gkatzelis2021fair,barman2022universal,gorokh2021monetary,vardi2021dynamic,benade2022dynamic,benade2023fair}. Our point of departure is the fact that we have two-sided constraints: we must guarantee efficiency to the drivers and fairness to food banks. A related, recent thread in fair division studies two-sided~\cite{freeman2021two,gollapudi2020almost,igarashi2022fair,patro2020fairrec} and allocator's preferences~\cite{bu2023fair}; a major difference with our work is that we have to make matching decisions online.

Several works in Economics and Computer Science study or collaborate with real-world food rescue organizations. For example,~\cite{prendergast2017food} discusses the transition to the market-based solution that Feeding America uses to allocate food to food banks.~\cite{lee2019webuildai,kahng2019statistical} develop a recommendation system to assist a non-profit organization in matching food donations to recipients.
~\cite{aleksandrov2015online} develop a simple algorithm that they use to help Food Bank Australia collect and deliver donated food.
~\cite{benade2023achieving} partner with a food rescue platform and develop a framework for more equitable distribution of donations.

%% file: prelims.tex
\section{Preliminaries}

We consider an undirected graph $G = (V,E)$, with $n = |V|$ nodes, corresponding to counties. A subset of these nodes, denoted by $\calF \subseteq V$, have food banks. For every node $v \in V$ there are $p_v \geq 0$ individuals that are food insecure and reside in that node. The total population is $N = \sum_{v \in V} p_v$. 
Each edge $e = (u,v) \in E$ has a weight $d(u,v)$ corresponding to the distance between $u$ and $v$; we assume that the distances satisfy the triangle inequality. The residents of each node are served from their closest food bank node. Let $f_{v} \in \mathcal{F}$ be the food bank that is closest to node $v \in V$, and let $N_{f} = \sum_{v \in V: f_{v} = f} p_v$ be the total number of individuals that food bank $f \in \mathcal{F}$ serves.

Drivers with food donations/loads arrive over time, one per round, for a total of $m$ rounds. A food load $\ell \in V^2 \times \mathbb{R}_{\geq 0}$ is defined as a triplet $\ell = (x, y, w)$ where $x \in V$ is the \emph{origin} of the driver, $y \in V$ is the \emph{destination} of the driver, and $w \geq 0$ is the \emph{value} of the food items that the driver is carrying (i.e., the value of the donation). Each food donation $\ell = (x,y,w)$ must be matched with a food bank $f \in \calF$ immediately and irrevocably upon its arrival.

\paragraph{Efficiency and Fairness.} Our goal is twofold: we want to minimize the distance drivers travel, while at the same time, we want the donations to be distributed fairly among the food banks.

We say that a matching algorithm is \emph{$\delta$-driver efficient} if, for all food donations $\ell = (x,y,w)$ matched to a food bank $f_{\ell} \in \calF$, it holds that $d(x,f_{\ell})+d(f_{\ell}, y)\le \delta \min_{f' \in \calF}{\{d(x,f')+d(f', y)\}}$.  That is, no driver travels more than a multiplicative $\delta$ factor of their optimal route, noting that this optimal route also passes through \emph{some} food bank.

We assume that food donations that a food bank receives are equally divided among the individuals it serves, and that each individual only cares about the total value of the food they receive. That is, we assume that individuals have identical preferences for food. Ideally, we would like an \emph{envy-free} allocation, i.e. every individual values the food they received more than they value the food a different person received. Let $w_f(t)$ be the total value of the food received by a food bank $f \in \calF$ at time $t$. We say that a matching algorithm is \emph{$\varepsilon$-multiplicatively envy-free}, or $\varepsilon$-mEF, if for all $f \in \calF$, $\varepsilon \, \frac{w_{f}(m)}{N_{f}} \geq \max_{f' \in \calF}{\left\{\frac{w_{f'}(m)}{N_{f'}}\right\}}$. That is, the total value received by an individual served by food bank $f$ at the end of the process (the total value $w_{f}(m)$ received by the food bank, divided by the population $N_f$ served by the food bank) is at least $1/\varepsilon$ times the value received by an individual served by a different food bank $f'$. For $\varepsilon=1$, the matching algorithm always outputs envy-free allocations, while, in general, $\varepsilon > 1$. We note that we add the quantifier ``multiplicative'' since the error in the standard notion of approximate envy-freeness is additive.

\paragraph{Stochastic arrivals, destinations, and values.} 
We assume that the parameters of each donation are stochastic. Let $\calD$ and $\calW$ be two probability distributions supported over $V$ and $\mathbb{R}_{\geq 0}$ respectively. We assume that for each incoming food donation $\ell = (x,y,w)$ the nodes $x,y$ are sampled independently from $\calD$, and the value $w$ of the donation is sampled from $\calW$. 

By considering worst-case choices for $\calD$ it is not difficult to convince ourselves that we can construct highly adversarial scenarios where no matching algorithm can simultaneously achieve non-trivial driver efficiency and non-trivial fairness. 
We thus make the following assumption on $\calD$. We say that the distribution $\calD$ is $(\alpha, \beta)$-biased if for every node $v \in V$ it holds that $\frac{p_v}{\alpha \cdot N} \le \Prob{u \sim \calD}{u=v} \le \frac{\beta \cdot p_v}{N}$. That is, a node $v$ is selected as the origin/destination node with probability almost proportional to $p_v/N$, the percentage of the population that resides in $v$. This assumption is reasonable, since in practice we expect most drivers to have their loads rejected at counties with a larger population (since those have more stores), and a larger population is also closely related to a larger food-insecure population.

Our positive results hold under the following less restrictive assumption. Let $S_{f} = \{v: v \in V, f_v = f\}$ be the set of nodes served by a food bank $f \in \calF$, and let $\calD'$ be the probability distribution over sets $S = \{S_{f}: f \in \calF\}$ induced by $\calD$ (i.e. $\calD$ samples a node, and $\calD'$ samples the set $S_{f}$ that contains that node). Then, for any $f \in \calF$, $\frac{N_{f}}{\alpha \cdot N} \le \Prob{S_{f^*} \sim \calD'}{S_{f^*}=S_{f}} \le \frac{\beta \cdot N_{f}}{N}$. The first assumption implies the second, but not vice-versa. Thus the second assumption is more general, and it will also be more convenient for our analysis.

%% file: theory.tex
\section{Dynamic Fair Division}\label{sec: theory}

In this section, we give a tight characterization of the fairness-efficiency trade-offs in our dynamic fair division problem.
Specifically, in~\Cref{sec: lower bound} we prove that $(3-\epsilon)$-driver efficiency is incompatible with non-trivial fairness guarantees.
In~\Cref{sec: hungrier or two} we present a tight upper bound: we give a $3$-driver efficient and $(1+\varepsilon)$-multiplicatively envy-free algorithm.
Our analysis crucially relies on a new balls-into-bins result, that we present separately in~\Cref{sec: balls and bins}.

\subsection{Lower bounds}\label{sec: lower bound}

We first show limits on the fairness guarantees of efficient matching algorithms.

\begin{theorem}\label{thm: lower bound}
Even for $(1,1)$-biased distributions and unit-value donations, there is no matching algorithm that is $(3-\delta)$-driver efficient and $\varepsilon$-multiplicatively envy-free, where $\delta>0$ and $\varepsilon < n$.
\end{theorem}

\begin{proof}

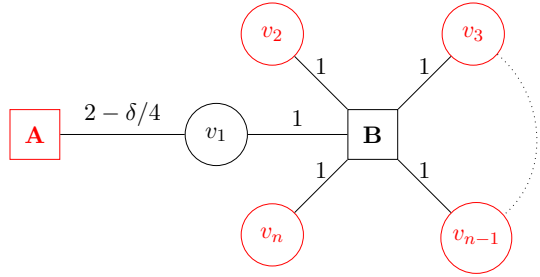
\begin{figure*}[h]
    \centering
    \resizebox{0.45\textwidth}{!}{
    \begin{tikzpicture}[node distance=2cm, auto]
        \node[rectangle, minimum size=8mm, draw, red,  font=\bfseries] (A) {A};
        \node[circle, minimum size=10mm,draw,  right= 20mm of A] (B) {$v_1$};
        \node[rectangle, minimum size=8mm, draw, right= 16mm of B, font=\bfseries] (C) {B};
        \node[circle, minimum size=10mm, draw, above left= 12 mm of C, red] (D) {$v_2$};
        \node[circle, minimum size=10mm, draw, above right= 12 mm of C, red] (E) {$v_3$};
        \node[circle, minimum size=10mm, draw, below right= 12mm of C, red] (F) {$v_{n-1}$};
        \node[circle, minimum size=10mm,  draw, below left= 12mm of C, red] (G) {$v_{n}$};
        
        \draw (A) -- node[above] {$2-\delta/4$} (B);
        \draw (B) -- node[above] {$1$} (C);
        \draw (C) -- node[above] {$1$} (D);
        \draw (C) -- node[above] {$1$} (E);
        \draw (C) -- node[above] {$1$} (F);
        \draw (C) -- node[above] {$1$} (G);
        \draw[dotted] (E) to[out=-40,in=40,looseness=1] (F);
    \end{tikzpicture}
    }
    \caption{An example of a hard instance. Nodes colored red, the set $\{A, v_2, \cdots, v_{n}\}$, are the locations of the food banks. The square nodes ($A$ and $B$) have zero population; all other nodes have a population of $1$.}
    \label{figure:counterExample}
\end{figure*}

    Consider the example in~\Cref{figure:counterExample}, with $V = \{A, B, v_1, \cdots, v_{n}\}$ and $\calF = \{A, v_2, \cdots, v_{n}\}$. Nodes $A, B$ have zero population, i.e., $p_A = p_B = 0$, while nodes $v_1, \cdots, v_{n}$ have population one, i.e., $p_{v_1} = \cdots = p_{v_{n}} = 1$. Therefore, $N = n-2$. We assume that all donations have unit value, i.e. $\calW$ takes the value $1$ with probability $1$, and that $\calD$ is $(1,1)$ biased, i.e. the origin/destination node is $v$ with probability exactly $p_v/N$. Finally, The distances between nodes are shown in~\Cref{figure:counterExample}, i.e., $d(A, v_1) = 2-\frac{\delta}{4}$ and $\forall i \in [n]: d(v_i, B) = 1$.

    Consider a $(3-\delta)$-driver efficient algorithm $\calM$. When $\calM$ delivers a donation to node $A$, it must be that both the origin node and the destination node of that donation were $v_1$, noting that $A$ and $B$ cannot be the origin/destination because they have zero population. 
    When $v_1$ is not both the origin and destination, then the origin or destination has a food bank, so the optimal distance is 2. Therefore, if we wanted to redirect the driver through $A$, the driving distance would be at least $6 - \frac{\delta}{2}$. Since $\calM$ is $(3-\delta)$-driver efficient this cannot be the case. 
    The probability with which $\calM$ serves $A$ is then at most $\frac{1}{n^2}$. The expected number of deliveries made to $A$ is therefore at most $\frac{m}{n^2}$. An envy-free allocation would give exactly $\frac{m}{n}$ deliveries to $A$. Thus, $\calM$ is at best $n$-mEF.
    \end{proof}

\subsection{The Power of Two Choices in Food Allocation}\label{sec: hungrier or two}

In this section we present our algorithm,~\Cref{ourAlgorithm}, as well as state and prove its guarantees.

\begin{algorithm}[t] 
         \SetAlgoLined
	     \KwIn{Graph $G = (V,E)$, distance metric $d(\cdot, \cdot)$, population $p_v$ for each node $v \in V$, set of food banks $\calF \in V$, an set of deliveries $\{\ell_j = (x_j, y_j, w_j)\}_{j \in [m]}$.}
	     \KwOut{An allocation of food deliveries to food banks.}
      
      \hrulefill
      
	   \For{each each time step $t$}{
	       1: Find the food banks $f_o$ and $f_d$ that are closest to the origin $x_t$ and the destination $y_t$, respectively. \\
            2: Assign the delivery to either $f_o$ or $f_d$, based on which one has received the least total value of food in the first $t-1$ steps, i.e., $f^* = \argmin_{f \in \{f_o, f_d\}} w_f(t-1)$.
	   }    
        \caption{Two Choices in Food Allocation}
        \label{ourAlgorithm}
\end{algorithm}

\begin{theorem}\label{thm: main result}
\Cref{ourAlgorithm} is $3$-driver efficient and $(1 + \varepsilon)$-multiplicatively envy-free with high probability, where
\begin{itemize}[leftmargin = *]
    \item If donations have unit value, and $2 \ln\left( \frac{\alpha \beta -1}{\alpha \beta - \beta} \right) > \ln\left( \frac{\alpha \beta -1}{\alpha -1} \right)$, then $\varepsilon \in O \left( \frac{N \log N}{m - N \log N} \right)$.
    \item If $\alpha < \sqrt{\frac{5}{4}}$, $\beta < \sqrt{\frac{4}{3}}$, $\EX{}{e^{\lambda \calW}} < \infty$, for some constant $\lambda$, and $\EX{}{\calW}=1$, then $\varepsilon \in O \left( \frac{N \log N}{m - N \sqrt{m}} \right)$.
\end{itemize}
\end{theorem}


\begin{proof}

    First, we prove the guarantee on driver efficiency. Let $\ell = (x,y,w)$ be a food donation, and $f_o \in \calF$ and $f_d \in \calF$ be the food banks closest to the origin and destination respectively. Let  $f^* \in \argmin_{f' \in \calF} \{ d(x,f')+d(f',y) \}$ be a food bank that minimizes the driving distance. We have
    \begin{align*}
        d(x,f_o) + d(f_o,y) &\le  d(x,f_o)+ d(f_o,x)+d(x,y) &(\text{triangle inequality}) \\
        &= 2 d(x,f_o) + d(x,y) & \\
        &\leq 2 d(x,f_o) + d(x,f^*) + d(f^*,y) &(\text{triangle inequality}) \\
        &\le 3 d(x,f^*) + d(f^*, y) & (d(x,f_o) \le d(x, f^*))\\
        &\le 3\left(d(x,f^*) + d(f^*, y)\right).
    \end{align*}
    Similarly, we can show that $d(x,f_d) + d(f_d,y) \le 3\left(d(x,f^*) + d(f^*, y)\right)$. This concludes the proof of the $3$-driver efficiency guarantee.

    The fairness guarantee is a direct implication of~\Cref{theorem:gapBound} and~\Cref{theorem:gapBoundWeightedBalls}, our balls-into-bins results. 
    Donations correspond to balls, with the value of a donation corresponding to the weight of the ball, while food banks correspond to bins, and $N_f$, the population a food bank serves, corresponds to the weight of a bin. At some step, the probability that a food bank is one of the two food banks considered by our algorithm is in $[\frac{N_{f}}{\alpha \cdot N}, \frac{\beta \cdot N_{f}}{N}]$; the probability that a bin is selected in balls-into-bins is the same. The balls-into-bins results give high probability bounds on the difference in value (defined as the total weight of balls in a bin, normalized by the bin's weight) between the heaviest and lightest bin, after throwing $m$ balls, i.e. bounds on $\max_{f \in |\calF|}\left \{ \frac{w_f(m)}{N_f} \right \} - \min_{f \in |\calF|} \left \{\frac{w_f(m)}{N_f} \right\}$, which imply bounds on $\varepsilon$-mEF.

    For example, assume that donations have unit value and use~\Cref{theorem:gapBound}. 
    For every $f \in \calF$ we have that, with high probability:
    \[ \max_{f' \in \calF}{\left\{\frac{w_{f'}(m)}{N_{f'}}\right\}} - \frac{w_{f}(m)}{N_{f}} \le \max_{f' \in |\calF|}\left \{ \frac{w_{f'}(t)}{N_{f'}} \right \} - \min_{f' \in |\calF|} \left \{\frac{w_{f'}(t)}{N_{f'}} \right\} \leq^{(\Cref{theorem:gapBound})} O(log N).\]
    Dividing by $\max_{f' \in \calF}{\left\{\frac{w_{f'}(m)}{N_{f'}}\right\}} > 0$ and rearranging we get that:
    \begin{align*}
        \frac{\frac{w_{f}(m)}{N_{f}}}{\max_{f' \in \calF}{\left\{\frac{w_{f'}(m)}{N_{f'}}\right\}}} \ge 1 -  \frac{O(log N)}{\max_{f' \in \calF}{\left\{\frac{w_{f'}(m)}{N_{f'}}\right\}}} \ge 1 - O\left( \frac{N\log N}{m} \right),
    \end{align*}
    were the last inequality holds due to the fact that  $\max_{f' \in \calF}{\left\{\frac{w_{f'}(m)}{N_{f'}}\right\}} \ge \frac{m}{N}$.
    Equivalently, we have $\frac{w_{f}(m)/ N_{f}}{\max_{f' \in \calF}{\left\{ w_{f'}(m)/ N_{f'} \right\}}} \geq \frac{1}{1 + \varepsilon}$, for $\varepsilon \in O \left( \frac{N \log N}{m - N \log N} \right)$.

    Similarly, using~\cref{theorem:gapBoundWeightedBalls} we can prove the same bounds for weighted balls, under the given assumptions.
    \begin{equation}
    \label{equation:Bound on envy}
      \max_{f' \in \calF}{\left\{\frac{w_{f'}(m)}{N_{f'}}\right\}} - \frac{w_{f}(m)}{N_{f}}  \leq^{(\cref{theorem:gapBoundWeightedBalls})} O(log N).  
    \end{equation}

    Since $M(\lambda) = \EX{}{e^{\lambda \calW}} < \infty$ for some $\lambda>0$, we have that for any $|z| < \lambda/2$, $M''(z) \le 2S$ for some constant $S \ge 1$. Thus $Var(\calW) = \EX{}{\calW^2}-\EX{}{\calW}^2 = M''(0)-1 \le 2S-1$.

    Let $X = \sum_{t=1}^m w(t)$ (were $w(t)$ is the weight of the ball at round $t$). Since the draws at each round are independent we have that $Var(X) = \sum_{t=1}^m Var(w(t)) \leq m(2S-1)$. From linearity of expectation, we also have that $\EX{}{X} = m$. 
    Chebyshev's inequality implies that $\Prob{}{X \le m - t} \le \frac{m(2S-1)}{t^2}$. Thus, by choosing $t= \sqrt{m(2S-1)N^k}$ we get that $\Prob{}{X \le m - N^{k/2}\sqrt{m(2S-1)}} \le N^{-k}$, and then, with high probability we have $X \ge m - N^{k/2}\sqrt{m(2S-1)}$. Taking $k=2$, $X \ge m - N\sqrt{m(2S-1)}$ with high probability. This implies that, with high probability, $\max_{f' \in \calF}{\left\{\frac{w_{f'}(m)}{N_{f'}}\right\}} \ge \frac{m-N\sqrt{m(2S-1)}}{N}$. 
    
    By rearranging \cref{equation:Bound on envy} and using the fact that $S$ is a constant, we get:

    \begin{align*}
        \frac{\frac{w_{f}(m)}{N_{f}}}{\max_{f' \in \calF}{\left\{\frac{w_{f'}(m)}{N_{f'}}\right\}}} \ge 1 -  \frac{O(log N)}{\max_{f' \in \calF}{\left\{\frac{w_{f'}(m)}{N_{f'}}\right\}}} \ge 1 - O\left( \frac{N\log N}{m - N\sqrt{m(2S-1)}} \right) = 1 - O\left( \frac{N\log N}{m - N\sqrt{m}} \right).
    \end{align*}
    This is equivalent to showing that $\frac{\frac{w_{f}(m)}{N_{f}}}{\max_{f' \in \calF}{\left\{\frac{w_{f'}(m)}{N_{f'}}\right\}}} \geq \frac{1}{1 + \varepsilon}$, for $\varepsilon \in O \left( \frac{N \log N}{m - N\sqrt{m} - N \log N} \right)$. Notice that the asymptotic behaviors are the same for both cases: as $m \rightarrow \infty$, $\varepsilon \rightarrow 0$.
\end{proof}

%% file: Majorization.tex
\section{Balls-into-bins}\label{sec: balls and bins}

Here, we present our balls-into-bins results. For the convenience of the reader who is solely interested in this set of results, as well as to streamline the usage of findings from the related work, we switch to balls-into-bins terminology for the remainder of this section. 

\subsection{Preliminaries}

There are $n$ bins and for each bin $i \in [n]$ we have a positive, integer weight $N_i$. Let $N = \sum_{i \in [n]} N_i$ be the total weight of bins. We are going to conduct the following experiment. For $m$ rounds we are going to sample independently two bins with replacement from a distribution $\calD$. $\calD$ is going to be $(\alpha, \beta)$-biased with respect to the weights $\{N_i\}_{i \in [n]}$, i.e., $\frac{N_{i}}{\alpha \cdot N} \le \Prob{i^* \sim \calD}{i^*=i} \le \frac{\beta \cdot N_{i}}{N}$. We are also going to sample the weight of the ball $w(t)$ at round $t$ from a probability distribution $\calW$. Throughout this section, we impose the mild restriction that there exists $\lambda$ such that $\EX{}{e^{\lambda \calW}} < \infty$. Also, for ease of analysis, and without loss of generality, we also assume that $\EX{}{\calW}=1$. 

The value of a ball of weight $w$ when allocated to bin $i$ is $\frac{w}{N_i}$. 
Let $w_i(t-1)$ denote the total weight of balls allocated to bin $i$ up until (and including) round $t-1$, and let $v_i(t-1) = \frac{w_i(t-1)}{N_i}$ denote the total value of balls allocated to bin $i$ up until (and including) round $t-1$. We allocate a fresh ball to one of two bins sampled independently from $\calD$, and specifically, we allocate the $t^{th}$ ball to the bin with the smaller (out of the two bins we sampled) $v_i(t-1)$ value. Let $Gap(t)$ denote the gap between the maximum and the minimum value at round $t$, i.e., $Gap(t) = \max_{i \in [n]} v_i(t) - \min_{j \in [n]} v_j(t)$.

In the analysis, it will be convenient to imagine that a bin of weight $N_i$ is comprised of $N_i$ unit slots. Let $s_{i,j}$ be the $j^{th}$ slot in bin $i$. When a ball of weight $w(t)$ arrives, the weight of that ball is equally distributed among all slots of the selected bin. Thus, the value of a bin is equal to the total weight allocated to any one of its slots. 
We can equivalently think of the overall process as first choosing two slots, where slot $s_{i,j}$ is picked probability $\frac{1}{N_i}\Prob{i^* \sim \calD}{i^*=i} \in [\frac{1}{\alpha \cdot N}, \frac{\beta}{N}]$, select the least loaded of the two slots, and then equally distribute the weight of the ball among the slots that belong to the same bin. Let $\psi_i(t)$ be the probability that, at round $t$, we select one of the $i$ most loaded slots and add the ball to the corresponding bin.

\subsection{Main Results}

We prove two theorems. The first result,~\Cref{theorem:gapBound}, only works for unit-weight balls. The second result,~\Cref{theorem:gapBoundWeightedBalls}, is our most general result. The analysis for the unit-weight paradigm is both simpler and tighter. However, in both cases, we ultimately provide the same asymptotic guarantees.

\begin{theorem}
\label{theorem:gapBound}
For weighted bins of a total weight of $N$, an $(\alpha, \beta)$-biased probability distribution $\calD$ such that $0< \epsilon \le \frac{2\cdot \ln\left( \frac{\alpha \beta -1}{\alpha \beta - \beta} \right)}{\ln\left( \frac{\alpha \beta -1}{\alpha -1} \right)}-1$, and unit weight balls, there exist constants $a$ and $c$ that only depend on $\epsilon$, such that for any constant $k>0$ and time $t \ge 0$: $\Prob{}{Gap(t) \ge \frac{k}{a} \, \log\left( \frac{c}{2a} \, N \right) } \le N^{-(k-1)}$.
\end{theorem}

Intuitively,~\Cref{theorem:gapBound} states that the difference in value between the most loaded bin and the least loaded bin (recalling that the value is total weight normalized by the weight of the bin) is at most logarithmic in $N$, the total weight of bins, and independent of $m$, the number of balls. Theorem 1.3 of~\cite{wieder2007Heterogeneous} implies that our analysis is tight in terms of $\alpha$ and $\beta$. That is, if $2 \ln\left( \frac{\alpha \beta -1}{\alpha \beta - \beta} \right) < \ln\left( \frac{\alpha \beta -1}{\alpha -1} \right)$, then there exist weights for the bins such that with high probability $Gap(t)$ will depend on the number of balls $m$.

Our more general theorem is stated as follows.

\begin{theorem}
\label{theorem:gapBoundWeightedBalls}
For weighted bins of a total weight of $N$, an $(\alpha, \beta)$-biased probability distribution $\calD$ such that $\alpha < \sqrt{\frac{5}{4}}$ and that $\beta < \sqrt{\frac{4}{3}}$, and weighted balls, with weight sampled from a distribution $\calW$ such that $\EX{}{e^{\lambda \calW}} < \infty$ and $\EX{}{\calW}=1$, there exist constants $a$, $\tilde{c}$ that depends only on $\calW$, $\alpha$, and $\beta$, such that for any constant $k>0$ and time $t \ge 0$: $\Prob{}{Gap(t) \ge \frac{2k}{a} \cdot \log\left( \frac{\tilde{c}}{a^{1/k}} \cdot N \right) } \le N^{-(k-1)}$.
\end{theorem}

\subsection{Unit-sized balls and the proof of~\Cref{theorem:gapBound}}

We start by proving~\Cref{theorem:gapBound}. The following process will be useful in our analysis.

\begin{definition}[$(1+\epsilon)$-choice with $N$ bins]
$m$ unit-weight balls are thrown into $N$ unit-sized bins. At each step, each ball is allocated to the least loaded of ``$1+\epsilon$'' bins selected uniformly at random. That is, each ball is added to one of the $i$ most loaded bins with probability $\phi_i \coloneqq \left(\frac{i}{N}\right)^{1+\epsilon}$. Let $Gap_{1+\epsilon}(t) \coloneqq \max_{i \in [N]} w_i(t)- \min_{i \in [N]} w_i(t) $ be the difference in weights at round $t$ between the most loaded and the least loaded bin of this process.
\end{definition}

In~\Cref{lemma:majorization} we prove that the gap in the $(1+\epsilon)$-choice with $N$ bins is at least the gap of our process. Before we state the lemma, we define the notions of majorization and stochastic dominance.

\begin{definition}[Majorization]\label{dfn: majorization}
        For two vectors $x, y$ each with $n$ elements we say that $y$ is majorized by $x$, denoted by $x \succeq y$, if $\forall j \le n$ it holds that $\sum_{i=1}^j x_{\pi(i)} \ge \sum_{i=1}^j y_{\pi'(i)}$, where $\pi$ and $\pi'$ are permutations such that $\forall i,j: i \le j$, $x_{\pi(i)} \ge x_{\pi(j)}$ and $y_{\pi'(i)} \ge y_{\pi'(j)}$.
    \end{definition}

\begin{definition}[Stochastic Dominance]
    For two random variables $X$ and $Y$, we say that $X$ stochastically dominates $Y$, denoted by $X \succeq Y$ if $\forall a$ it holds that $\Prob{}{X \le a} \le \Prob{}{Y \le a}$.
\end{definition}

\begin{claim}
    \label{claim:splitting majorization}
    Consider two vectors $y \in \R^n$ and $x \in \R^n$ and a set $S \subseteq [n]$, such that $\forall i \notin S, x_i = y_i$, and $\forall i \in S, y_i = \frac{1}{|S|}\sum_{i \in S} x_i$. Then $y$ is majorized by $x$ (i.e. $x \succeq y$)
\end{claim}

\begin{proof}
    It is a well-known fact (e.g.,~\cite{marshall1979inequalities}) that a vector $x$ majorizes a vector $y$ if and only if there exists a doubly stochastic matrix $P$ such that $y = Px$. Consider the following matrix $P$. For all $i \notin S$ we have that $P_{i, i} = 1$. For all combinations $i, j \in S \times S$  we have that $P_{i, j} = \frac{1}{|S|}$. Everywhere else the matrix is zero. By construction, $P$ is doubly stochastic. Furthermore, it is easy to verify that, by definition $y = P x$.
\end{proof}

\begin{lemma}
\label{lemma:majorization}
If $0 < \epsilon \le \frac{2\cdot \ln\left( \frac{\alpha \beta -1}{\alpha \beta - \beta} \right)}{\ln\left( \frac{\alpha \beta -1}{\alpha -1} \right)}-1$ then $Gap_{1+\epsilon}(t) \succeq Gap(t) $.
\end{lemma}

\begin{proof}
    Since $\calD$ is $(\alpha, \beta)$-biased then the sampling probability of \emph{slots} is also $(\alpha, \beta)$-biased. Recall that $\psi_i(t)$ is the probability that, at round $t$, we select one of the $i$ most loaded slots and add the ball to the corresponding bin. This is simply the probability that two $(\alpha, \beta)$-biased draws from $\calD$ both pick one of the $i$-th most loaded slots. Thus, we can use the following claim from~\cite{wieder2007balanced} (slightly rephrased for our purposes) that analyzes $(\alpha, \beta)$-biased $d$-choice processes:
    \begin{claim}[Special case of Claim 2.5 of~\cite{wieder2007balanced}]
    For every $k > 0$, if $\epsilon \le \frac{2\cdot \ln\left( \frac{\alpha \beta -1}{\alpha \beta - \beta} \right)}{\ln\left( \frac{\alpha \beta -1}{\alpha -1} \right)}-1$, then for every $i \leq N$ and $t \geq 0$ it holds that $\psi_i(t) \le \phi_i$.
    \end{claim}

    Note that although our process differs from the one analyzed in \cite{wieder2007balanced} we can readily use this claim. The claim compares the selection probability of bins in two different scenarios. In terms of selection probability bins and slots are equivalent since the way the ball is distributed after the insertion does not affect the initial selection probabilities.

    Let $s(t)$ denote the sorted vector of slot weights at round $t$ for our process and $x(t)$ denote the sorted weight vector of bins for the $(1+\epsilon)$-choice process. The allocation probabilities of each process define Markov chains $\left(s(t)\right)_{t}$ and $\left(x(t)\right)_{t}$. We prove our result by showing that there exists a coupling between the two Markov chains for which, at every round $t$, $x(t) \succeq s(t)$.

    We prove the claim by induction. The fact that $x(0) \succeq s(0)$ is trivial. Assume that $x(t-1) \succeq s(t-1)$. We will show how to couple the bin sampling between the two processes. First, pick a number $\lambda$ uniformly at random from $[0,1)$. Let $i$ be such that $\phi_{i-1} \le \lambda \le \phi_i$ and $j$ such that $\psi_{j-1}(t) \le \lambda \le \psi_{j}(t)$. We are going to choose the $i^{th}$ bin in the $(1+\epsilon)$-choice process, and the $j^{th}$ slot in our process. It is easy to verify that this is indeed a valid coupling since the marginal selection probabilities respect the selection probabilities of the individual processes. Now, let $e_i$ be the vector with $1$ the $i^{th}$ coordinate and 0 everywhere else. First we will show that $x(t-1) +  e_i  \succeq s(t-1) + e_j$. From our induction hypothesis, we have that $x(t-1)  \succeq s(t-1)$. Furthermore, since $\psi_i(t) \le \phi_i$ for every $i \le N$, we can conclude that $j \ge i$. Now, notice that $x(t-1) +  e_i  \succeq s(t-1) + e_j$ is an immediate implication of the following claim of~\cite{wieder2007balanced}.
    \begin{claim}[Claim 2.4 in~\cite{wieder2007balanced}]
    For two sorted integer vectors $x, y$ such that $x \succeq y$, if $i \le j$, then $x +  e_i  \succeq y + e_j$, where the vectors $x + e_i, y + e_j$ are also sorted.
    \end{claim}
    Using \cref{claim:splitting majorization} on $s(t-1) + e_j$ and $s(t)$, and from transitivity of majorization we conclude that $x(t) \succeq s(t)$ (which concludes our induction). Finally, observe that since $x(t) \succeq s(t)$ and $\left\lVert x(t) \right \rVert_1 = \left\lVert s(t) \right \rVert_1$, from the definition of majorization we get that $\max_i x_i(t) \ge \max_i s_i(t)$ and $\min_i x_i(t) \le \min_i s_i(t)$. Thus, for our proposed coupling  $Gap(t) \le Gap_{1+\epsilon}(t)$ which implies that $Gap_{1+\epsilon}(t) \succeq Gap(t)$,  concluding our proof.
\end{proof}

Now we are ready to prove the main theorem of this section.

\begin{proof}[Proof of~\Cref{theorem:gapBound}]
\Cref{lemma:majorization} implies that we only need to bound the $(1+\epsilon)$-choice process with $N$ unit-sized bins.
Let $x(t)$ denote the sorted weight vector of bins at round $t$ for this process. Notice that
\[
Gap_{1+\epsilon}(t) = \max_{i} x_i(t) - \min_i x_i(t) =\max_{i} x_i(t) - t/N + t/N - \min_i x_i(t) \le 2 \max_i |x_i(t) - t/N|.
\]
This implies that 
\begin{equation}\label{eq: useful eq 1}
  \EX{}{e^{a \cdot Gap_{1+\epsilon}(t)}} \le \EX{}{e^{2 a\cdot \max_i |x_i(t) - t/N|}} \le \EX{}{\sum_{i=1}^N e^{2a| x_i(t) - t/N|}}.
\end{equation}

We use the following result of~\cite{yuval2015}:

\begin{theorem}[Theorem 2.2 in~\cite{yuval2015}, rephrased]\label{thm: yuval}
    For any $t \ge 0$, $\EX{}{\Gamma(t)} \le \frac{4c'}{a' \epsilon'}n$, where $\Gamma(t)=\sum_{i=1}^N e^{a'| x_i(t) - t/N|}$, where $a'$ is a hyperparameter of our choosing that is upper bounded by $\frac{\epsilon'}{10}$, $\epsilon'$ is a constant that depends on $\epsilon$, and $c'$ is a constant.
\end{theorem}

We note that the $(1+\beta)$-process examined in \cite{yuval2015} is not the same as the $1+\epsilon$ choice algorithm presented here. Nevertheless, one key observation is that their results hold for processes that satisfy the following conditions. Let $\phi'_i$ be the probability that a ball is added to one of the $i$ most loaded bins. Then, if $\forall i \in [1, N-1], \phi'_{i}-\phi'_{i-1} \le \phi'_{i+1}-\phi'_{i}$ (where $\phi'_0=0$), $\phi'_{N/4} < 1/4$, and $\phi'_{3N/4} < 3/4$, their results go through. For any $\epsilon>0$, the $1+\epsilon$ choice process satisfies these constraints, and therefore~\Cref{thm: yuval} can be readily used.
Setting $c=(\frac{4c'}{\epsilon'})^{1/k}$ and $a = (a')^{1/k}/2$ for a constant $k$, we get that $\EX{}{\sum_{i=1}^N e^{2a| x_i(t) - t/N|}} \le \frac{c^k}{(2a)^k}N$. We have:
    \begin{align*}
        &\Prob{}{Gap(t) \ge \frac{k}{a} \cdot \log \left( \frac{c}{2a} N \right)} \\&\le \Prob{}{Gap_{1+\epsilon}(t) \ge \frac{k}{a} \cdot \log \left( \frac{c}{2a} N \right)} \tag{\cref{lemma:majorization}}\\
         &= \Prob{}{aGap_{1+\epsilon}(t) \ge \log \left( \frac{c^k}{(2a)^k} N^{k} \right)} \\
        &= \Prob{}{e^{aGap_{1+\epsilon}(t)} \ge N^{k-1} \frac{c^k}{(2a)^k} N} \tag{Monotonicity of $\exp(\cdot)$}\\
        &\le \Prob{}{e^{aGap_{1+\epsilon}(t)} \ge  N^{k-1} \: \EX{}{\sum_{i=1}^N e^{2a| x_i(t) - t/N|}}}\tag{\Cref{thm: yuval}}\\
        &\le \Prob{}{e^{aGap_{1+\epsilon}(t)} \ge  N^{k-1} \: \EX{}{e^{aGap_{1+\epsilon}(t)}}}\tag{\Cref{eq: useful eq 1}}\\
        &\le N^{-(k-1)}. \tag{Markov's inequality}
    \end{align*}

\end{proof}

%% file: WeightedBalls.tex
\subsection{Weighted balls and the proof of~\cref{theorem:gapBoundWeightedBalls}}

The main issue with extending the above analysis to the case of weighted balls is that the majorization argument of~\Cref{lemma:majorization} does not go through. Consider a state where one process has equally distributed the balls so far among all $n$ bins, and one that has equally distributed the balls among $n-1$ bins. Obviously, the second vector majorizes the first one. Now, say that a very heavy ball arrives. The second process might allocate it to the empty bin, whereas the first process has to allocate it to some bin that already has a significant load. After that insertion, if the weight of the ball is big enough, the second process will not majorize the first one anymore. This example demonstrates why coupling and majorizing do not seem like viable options for weighted balls. Instead, we use a potential function argument, similar to the one introduced in~\cite{yuval2015}.

Throughout the analysis we assume that $\alpha < \sqrt{\frac{5}{4}}$ and that $\beta < \sqrt{\frac{4}{3}}$. Equivalently, there exists a constant $\mu$ such that $\frac{1}{\alpha^2} \ge \frac{4}{5} +\mu$ and $\beta^2 \le \frac{4}{3} - \mu$.
 
 Consider the vector $x(t)$ which is defined as $x_i(t) \coloneqq \frac{1}{N_i}\left(w_i(t) - \frac{N_i}{N}\sum_{t' \le t} w(t')\right) = v_i(t) - \frac{1}{N}\sum_{t' \le t} w(t')$, where $w(t')$ is the weight of the ball at round $t'$. Essentially $x_i(t)$ captures the distance of the value of a bin from its value if everything was perfectly distributed. For the remainder of the analysis, and without loss of generality, we assume that $x_i(t)$ is sorted. Define the following potential functions, where $0 < a < 1$ is a small constant:
 \begin{enumerate}
    \item $\Phi(x(t)) = \sum_{i=1}^n e^{a x_i(t)}$.
    \item $\Psi(x(t)) = \sum_{i=1}^n e^{-a x_i(t)}$.
    \item $\Gamma(x(t)) = \Phi(x(t)) + \Psi(x(t))$.
\end{enumerate}
For the sake of the analysis, we will again consider slots. Each bin $i$ of weight $N_i$ is composed of a total of $N_i$ unit-sized slots. Let $x^s(t)$ be the normalized slot vector. For a bin $i \in [n]$ and slot $j \in [N_i]$ we have that $x^s_{i,j}(t) \coloneqq v_i(t) - \frac{1}{N}\sum_{t' \le t} w(t') = x_i(t)$ which captures the distance of the weight of the slot from the weight it would have if everything was perfectly distributed.

Our final goal is to compute a bound on $\EX{}{\Gamma(x(t))}$. Specifically, the majority of our analysis is spent towards proving the following lemma:

\begin{restatable}{lemma}{lemmaGammabound}
\label{lemma:GammaBound}
    For any $t\ge 0$, $\EX{}{\Gamma(x(t))} \le \frac{c^*}{a \zeta} N$, where $\zeta = \min{\{\frac{\mu}{4}, \frac{1}{60}\}}$ and $c^*$ is a constant that only depends on $a$ and $\mu$.
\end{restatable}

Once we have established~\Cref{lemma:GammaBound}, we can use arguments similar to the ones in the proof of~\Cref{theorem:gapBound} to get the desired results:

\begin{proof}[Proof of~\Cref{theorem:gapBoundWeightedBalls}]

As in the proof of \cref{theorem:gapBound} notice that $Gap(t) \le 2 \cdot \max_i |v_i(t) - t/N|$ and thus $e^{\frac{a}{2}Gap(t)} \le  \sum_{i=1}^n e^{a |x_i(t)|} \le \Gamma(x(t))$.
Setting $\tilde{c}= \left(\frac{c^*}{\zeta}\right)^{1/k}$ for some constant $k$ we have:
    \begin{align*}
        \Prob{}{Gap(t) \ge \frac{2k}{a} \cdot \log \left( \frac{\tilde{c}}{a^{1/k}} N \right)} &= \Prob{}{\frac{a}{2}Gap(t) \ge \log \left( \frac{\tilde{c}^k}{a} N^{k} \right)} \\
        &= \Prob{}{e^{\frac{a}{2}Gap(t)} \ge N^{k-1}\frac{\tilde{c}^k}{a} N} \\
        &\le \Prob{}{e^{\frac{a}{2}Gap(t)} \ge  N^{k-1} \: \EX{}{\Gamma(x(t))}} &(\cref{lemma:GammaBound})\\
        &\le \Prob{}{\Gamma(x(t)) \ge  N^{k-1} \: \EX{}{\Gamma(x(t))}}\\
        &\le N^{-(k-1)}. &(\text{Markov's Inequality})
    \end{align*}
\end{proof}

We are going to prove a series of lemmas that bound the expected increase of $\Phi(x(t))$ and $\Psi(x(t))$ under mutually exclusive scenarios. By combining these lemmas we will have a bound on the expected increase of $\Gamma(x(t))$ which, in turn, will allow us to establish~\Cref{lemma:GammaBound}. 
As opposed to previous results in the balls-into-bins process, our analysis needs to account for the weighted balls, the weighted bins, and the $(\alpha, \beta)$-biased distribution. To take care of the fact that balls are weighted, we borrow from the potential function approach of~\cite{yuval2015}. To handle weighted bins we again resort to slots: instead of bounding the potential functions for the bin vector, we are going to bound it for the slot vector. In turn, this will produce bounds on the desired quantity. To bound the expected increase of the potential function for the slot vector we are going to use a weaker intermediate process that adds the ball to the selected slot, but does not equally distribute the load among the slots of the corresponding bin. Finally, regarding the $(\alpha, \beta)$-biased sampling, we prove upper and lower bounds on the probability we select the $i^{th}$ most loaded slot; we can delicately plug these bounds in our analysis without occurring a super-constant loss.

Starting with defining the potential function for the slot vectors, let $x_i^{s}(t)$ denote the $i^{th}$ coordinates of the slot vector for all $i \in [N]$. It is not difficult to see that:
\[\Phi(x(t)) = \sum_{i=1}^n e^{a x_i(t)} \le \sum_{i=1}^n \sum_{j=1}^{N_i} e^{a x_i(t)} = \sum_{i=1}^n \sum_{j=1}^{N_i} e^{a x^s_{i,j}(t)} = \Phi(x^s(t)).\]

Similarly we can prove that $\Psi(x(t)) \le \Psi(x^s(t))$ and thus $\Gamma(x(t)) \le \Gamma(x^s(t))$. The following analysis aims to upper bound $\EX{}{\Gamma(x^s(t))}$ which will immediately give a bound on $\EX{}{\Gamma(x(t))}$.

Let $p_{i}(t)$ be the probability with which we select the $i^{th}$ most loaded slot in the normalized slot vector, and add the ball to the respective bin. For ease of notation, we omit dependence on time and simply write $p_{i}$ when $t$ is implied by the context. Recall that the probability distribution from which slots are sampled is $(\alpha, \beta)$-biased. If $(\alpha, \beta) = (1,1)$ (i.e. slots were sampled from the uniform distribution) it is not difficult to see that $p^{U}_{i} = \left( \frac{i}{N}\right)^2 - \left( \frac{i-1}{N} \right)^2 = \frac{2i-1}{N^2}$. We prove that $p_i$ can be approximated by $p^{U}_{i}$.

\begin{lemma}
\label{lemma:abBound}
    Let $\calD$ be an $(\alpha, \beta)$-biased distribution for sampling slots. Then for any $i \in [n]$:
    \[\frac{1}{\alpha^2}\cdot p^{U}_i \le p_i \le \beta^2 \cdot p^{U}_i.\]
\end{lemma}

\begin{proof}
    We have:
\[ p_i = p_i^2 + 2 \, p_i \,\sum_{j < i} p_j  \le  \left(\beta \frac{1}{N}\right)^2 + 2 \left(\beta \frac{1}{N}\right) \sum_{j < i} \left(\beta \frac{1}{N}\right) = \beta^2 \cdot p^{U}_{i}.\]
A near identical argument shows that $p_i \geq \frac{1}{\alpha^2} \cdot p^{U}_{i}$.
\end{proof}

We have assumed that $M(\lambda) = \EX{ w \sim \calW }{e^{\lambda w}} < \infty$, which implies that for any $|z| < \lambda/2$, $M''(z) \le 2S$ for some constant $S \ge 1$. We first bound the difference in $\Phi(.)$ between consecutive steps.

\begin{lemma}\label{lemma: first phi bound}
$\EX{}{\Phi(x^s(t+1))- \Phi(x^s(t)) | x^s(t)} \le \sum_{i =1}^N \left( p_{i}(a + Sa^2)-(\frac{a}{N} - S\frac{a^2}{N^2}) \right) e^{a x^s_i(t)}$.
\end{lemma}

\begin{proof}
Throughout the analysis, we assume an appropriately small, yet constant choice of $a$.
 \begin{align*}
     &\EX{w(t+1) \sim \calW}{e^{a\left(x_i(t)+w(t+1)-\frac{w(t+1)}{N}\right)} - e^{ax_i(t)}} = e^{ax_i(t)} \cdot M( a ( 1- 1/N )) - e^{ax_i(t)}\\
     &\qquad= e^{ax_i(t)} \left(M(0) +M'(0) a(1-1/N) + M''(\zeta) \left( a(1-1/N)\right)^2/2 \right) - e^{ax_i(t)} \\
     &\qquad \le e^{ax_i(t)} \left(1 + a(1-1/N) + 2S \left( a(1-1/N)\right)^2/2 -1 \right) \\
     &\qquad \le e^{ax_i(t)} \left((1-1/N)a+Sa^2 \right),
 \end{align*}

 where the second equality comes from the Lagrange form of Taylor’s remainder theorem on $M(\lambda)$ around $\lambda = 0$, and $\zeta$ is some constant in $[0, (1-1/N)a]$. Similarly, we can show that:

\[\EX{{w(t+1)} \sim \calW}{e^{a\left(x_i(t)-\frac{w(t+1)}{N}\right)} - e^{ax_i(t)}} \le e^{ax_i(t)} \left(-\frac{a}{N}+S\frac{a^2}{N^2} \right).\]

    Now, given a slot vector $x^s(t)$ and a ball at step $t+1$, consider two processes. 
    One chooses a slot (based on the given probability distribution $\calD$) and equally splits the weight of the ball among the slots of the same bin; let $x^s(t+1)$ be the corresponding slot vector. The other process chooses a slot (based on the given probability distribution $\calD$) and assigns the entire weight of the ball to that slot. Let $x^*(t+1)$ be the resulting slot vector, and note that, if $i$ was the chosen slot, then $x^*(t+1) = x^s(t)+ w(t+1)\cdot e_i - \frac{w(t+1)}{N} \cdot \mathbf{1}$, where $e_i$ is the vector with $1$ in the $i$-th coordinate and $0$ everywhere else, and $\mathbf{1}$ is the all ones vector. 
    
    We first show that $\Phi(x^*(t+1)) \succeq \Phi(x^s(t+1))$ given $x^s(t)$. We couple the two processes, by making both of them choose the same slot $i$. Let $j$ be the bin to which $i$ corresponds. Under this coupling, from \cref{claim:splitting majorization} it is evident that $x^*(t+1)$ majorizes $x^s(t+1)$ (as per~\Cref{dfn: majorization}).  Thus $x^*(t+1) \succeq x^s(t+1)$. It is also well known (\cite{marshall1979inequalities}) that if $x \succeq y$ and $g$ is a convex function, it holds that $\sum_i g(x_i) \ge \sum_i g(y_i)$. Since $g(x) = e^{ax}$ is a convex function it is clear that $\Phi(x^*(t+1)) \ge \Phi(x^s(t+1))$ given $x^s(t)$. 
    Thus:
    \[\EX{}{\Phi(x^s(t+1))- \Phi(x^s(t)) | x^s(t)} \le \EX{}{\Phi(x^*(t+1))- \Phi(x^s(t)) | x^s(t)}.\]
    Now, observe that if we add a ball at slot $i$ during round $t+1$ then $x^*_i(t+1) = x^s_i(t) + w(t+1)-\frac{w(t+1)}{N}$, where $w(t+1)$ is the weight of the ball at round $t+1$ and is sampled from $\calW$. Otherwise, $x^*_i(t+1) = x^s_i(t) -\frac{w(t+1)}{N}$. Combining all of the above, we have:

    \begin{align*}
        &\EX{}{\Phi(x^s(t+1))- \Phi(x^s(t)) | x^s(t)} \le \EX{}{\Phi(x^*(t+1))- \Phi(x^s(t)) | x^s(t)} \\ &\le \sum_{i =1}^N  p_i\EX{w(t+1) \sim \calW}{e^{a\left(x^s_i(t)+w(t+1)-\frac{w(t+1)}{N}\right)} - e^{ax^s_i(t)}} + (1-p_i)\EX{w(t+1) \sim \calW}{e^{a\left(x^s_i(t)-\frac{w(t+1)}{N}\right)} - e^{ax^s_i(t)}}\\
        &\le \sum_{i =1}^N  \left(p_i\left( \left( 1 - \frac{1}{N} \right)a+Sa^2 \right) + (1-p_i)\left(-\frac{a}{N}+S\frac{a^2}{N^2} \right) \right)e^{ax^s_i(t)}\\
        &\le \sum_{i =1}^N \left( p_i(a + Sa^2)-\left(\frac{a}{N} - S\frac{a^2}{N^2}\right) \right) e^{a x_i(t)}. \qedhere
    \end{align*}
\end{proof}

The proof of the following corollary is deferred to~\Cref{app: missing proofs}.

\begin{corollary} \label{corollary:Phi}
    $\EX{}{\Phi(x^s(t+1))- \Phi(x^s(t)) | x^s(t)} \le \beta^2 \frac{2a}{N}\Phi(x^s(t))$.
\end{corollary}

The proof of the following lemma is identical to the proof of \cref{lemma:Phi_1}.

\begin{lemma}
$\EX{}{\Psi(x^s(t+1))- \Psi(x^s(t)) | x^s(t)} \le \sum_{i =1}^N \left( p_i(-a + Sa^2)+(\frac{a}{N} + S\frac{a^2}{N^2}) \right) e^{-a x^s_i(t)}$.
\end{lemma}

\begin{corollary}
\label{corollary:Psi}
    $\EX{}{\Psi(x^s(t+1))- \Psi(x^s(t)) | x^s(t)} \le \frac{3a}{2N}\Psi(x^s(t))$.
\end{corollary}

\begin{proof}
    Simply observe that $(-a + Sa^2) < 0$ and $p_i \ge 0$.
\end{proof}

Continuing, our next task is to prove bounds on the potential functions $\Phi(x^s(t))$ and $\Psi(x^s(t))$, as well as the difference in potential between two consecutive time steps, for the cases when $x^s_{\frac{3N}{4}}(t)$ is small and $x^s_{\frac{N}{4}}(t)$ is large. The next few lemmas, whose proofs are deferred to~\Cref{app: missing proofs}, establish these bounds.
Recall that $\mu$ is a constant such that $\frac{1}{\alpha^2} \ge \frac{4}{5} +\mu$ and $\beta^2 \le \frac{4}{3} - \mu$.

\begin{lemma} \label{lemma:Phi_1}
    If $x^s_{\frac{3N}{4}}(t) \le 0$, then $\EX{}{\Phi(x^s(t+1))|x^s(t)} \le (1-\frac{\mu a}{2N})\Phi(x^s(t))+1$.
\end{lemma}

\begin{lemma} \label{lemma:Psi_1}
    If $x^s_{\frac{N}{4}}(t) \ge 0$, then $\EX{}{\Psi(x^s(t+1))|x^s(t)} \le (1-\frac{3 \mu a}{4N})\Psi(x^s(t))+1$.
\end{lemma}

\begin{lemma} \label{lemma:Phi_2}
    Suppose that $x^s_{\frac{3N}{4}}(t)>0$ and $\EX{}{\Phi(x^s(t+1)) - \Phi(x^s(t))| x^s(t) } \ge -\frac{1}{36} \frac{a}{N}\Phi(x^s(t))$. Then either $\Phi(x^s(t)) < \frac{3 \epsilon}{32} \Psi(x^s(t))$ or $\Gamma(x^s(t)) < cN$ where $c$ is a constant that depends on $a$ and $\mu$.
\end{lemma}

\begin{lemma} \label{lemma:Psi_2}
    Suppose that $x^s_{\frac{N}{4}}(t)<0$ and $\EX{}{\Psi(x^s(t+1)) - \Psi(x^s(t))| x^s(t)} \ge - \frac{a}{60N}\Psi(x^s(t))$. Then either $\Psi(x^s(t)) < \frac{\mu}{6} \Phi(x^s(t))$ or $\Gamma(x^s(t)) < c'N$ where $c'$ is a constant that depends on $a$ and $\mu$.
\end{lemma}

Given these lemmas, we can bound $\Gamma(.)$.

\begin{lemma}\label{lemma: first bound on gamma}
   $ \EX{}{\Gamma(x^s(t+1))|x^s(t)} \le (1-\zeta \frac{a}{N})\Gamma(x^s(t)) + c^*$ where $\zeta = \min{\{\frac{\mu}{4}, \frac{1}{60}\}}$ and $c^* = \max{\{2, c, c'\}}$, where $c$ and $c'$ are the constants guaranteed in Lemmas~\ref{lemma:Phi_2} and~\ref{lemma:Psi_2}.
\end{lemma}

\begin{proof}

We consider various mutually exclusive cases, depending on the values of $x^s_{\frac{3N}{4}}(t)$ and $x^s_{\frac{N}{4}}(t)$.
For ease of notation, let $\Delta \Phi = \Phi(x^s(t+1)) - \Phi(x^s(t))$, $\Delta \Psi = \Psi(x^s(t+1)) - \Psi(x^s(t))$, and $\Delta \Gamma = \Gamma(x^s(t+1)) - \Gamma(x^s(t))$.

\textbf{Case 1: $x^s_{N/4}(t) \ge 0 \geq x^s_{\frac{3N}{4}}(t)$.}  From~\Cref{lemma:Phi_1} and~\Cref{lemma:Psi_1} we have that $\EX{}{\Delta \Gamma|x^s(t)} = \EX{}{\Delta \Phi|x^s(t)} + \EX{}{\Delta \Psi|x^s(t)} \le -\frac{\mu a}{2N}\Phi(x^s(t)) - \frac{3\epsilon a}{4N}\Psi(x^s(t)) +2 \le -\frac{\epsilon a}{2N}\Gamma(x^s(t)) +2$.

\textbf{Case 2: $x^s_{N/4}(t) \ge x^s_{\frac{3N}{4}}(t) \ge 0$.} If $\EX{}{\Delta \Phi|x^s} \le -\frac{1}{36} \frac{a}{N} \Phi(x^s(t))$, then~\cref{lemma:Psi_1} implies that $\EX{}{\Delta \Gamma|x^s} \le -\min{\{\frac{1}{36}, \frac{3\mu}{4}\}}\frac{a}{N}\Gamma+1$. Otherwise,~\Cref{lemma:Phi_2} implies that we have two distinct, exhaustive cases: (i) $\Phi(x^s(t)) < \frac{3 \cdot \mu}{32} \Psi(x^s(t))$, and (ii) $\Gamma(x^s(t)) < cN$. We consider those cases, as well as the remaining case that $0 \geq x^s_{N/4}(t) \ge x^s_{\frac{3N}{4}}(t)$, in~\Cref{app: missing proofs}.
\end{proof}

Finally, we are ready to prove~\Cref{lemma:GammaBound}.

\lemmaGammabound*
\begin{proof}[Proof  of~\Cref{lemma:GammaBound}]
    We will prove this by induction. For $t=0$ the statement is immediate.
\begin{align*}
    \EX{}{\Gamma(x^s(t+1))} &= \EX{\Gamma(x^s(t))}{\EX{}{\Gamma(x^s(t+1)) |\Gamma(x^s(t))}}\\
    &\leq \EX{}{\left(1-\zeta \frac{a}{N} \right) \Gamma(x^s(t)) + c^*} &(\Cref{lemma: first bound on gamma})\\
    &\le \left(1-\zeta \frac{a}{N} \right)\frac{c^*}{a \zeta} N + c^* & (\text{Induction Hypothesis})\\
    &=\frac{c^*}{a \zeta} N.
\end{align*}
Since $\Phi(x(t)) \le \Phi(x^s(t))$ and  $\Psi(x(t)) \le \Psi(x^s(t))$, $\Gamma(x(t)) \le \Gamma(x^s(t))$; the lemma follows.
\end{proof}

%% file: experiments.tex
\section{Experiments}\label{sec: experiments}

In this section, we describe our experiments. We compare~\Cref{ourAlgorithm} against a parameterized family of benchmarks. Our benchmarks are formally described below.

\begin{enumerate}[leftmargin=*]
    \item The driver optimal algorithm: given a food donation $\ell = (x,y,w)$, match it to the food bank $f$ that minimizes the driver's travel distance, i.e. $f \in \argmin_{f' \in \calF} \{ d(x,f') + d(f',y) \}$.
    \item The greedy algorithm: given a food donation $\ell = (x,y,w)$, allocate it to the food bank $f$ with the smallest normalized load i.e., $f \in \argmin_{f' \in \calF}{\frac{w_{f'}(t)}{N_{f'}}}$, irrespective of its location. Since food banks have identical preferences, it is easy to see that this algorithm always allocates a donation to an ``unenvied'' food bank, and no envy-cycle ever arises. Thus, the resulting allocation satisfies the (weighted) \emph{envy freeness up to one item} ($\mathrm{EF}1$) guarantee, a well-established fairness notion.
    \item Greedy with a cutoff of $c$ miles: given a food donation $\ell = (x,y,w)$, consider all the food banks $f$ such that the driver has to travel at most $c$ additional miles for the delivery compared to the ``driver optimal'' food bank. Let $\calF^c = \{f \in \calF: (d(x,f)+d(f, y))-\min_{f' \in \calF}{\{d(x,f')+d(f', y)\}} \le c \}$). Greedy with a cutoff of $c$ picks the food bank in $\calF^c$ that has the smallest normalized load. This procedure smoothly interpolates between the greedy algorithm and the driver optimal algorithm: greedy with a cutoff of $0$ miles is exactly the driver optimal algorithm, and greedy with a cutoff of $c$ miles, for a large enough $c$, is the standard greedy algorithm.
\end{enumerate}

\paragraph{Data sources and Experiment Details.} In our experiments we construct the underlying graph of each area using the Bing maps API. In particular, each county in a state is represented by a node on a graph, and the edge weights, i.e., the distances between the counties, are computed via the map data. For each county, we use Feeding America's Map the Meal Gap 2023 dataset~\cite{MapMeal23} to get the food insecure and the total populations. Additionally, from the same dataset, we use the location of Feeding America's partner food banks~\cite{FeedAmerica24}. In each experiment, we sample the source and destination counties proportionally to the total population of the county (not the food insecure population) and the value of a donation is sampled from an exponential distribution with a mean of roughly $10,447/30 = 348$, to match the average amount of food (in pounds) that the Food Drop program receives daily, based on its current monthly average. In each of our experiments, we sample $50,000$ donations and our final results are obtained by performing $100$ such experiments and averaging the results. We compare the statistics of multiplicative envy and driver efficiency. Our code can be found in the supplementary material of this submission.

 We have tested~\Cref{ourAlgorithm} and benchmarks for all states with a sufficiently large number of food banks, and our results are consistent across the US. Here we present a representative set of four scenarios: the state of Indiana, Indiana++ (the area comprised of Indiana, Illinois, and Kentucky), the state of  California, and the state of Virginia. The first two scenarios (in~\Cref{subsec: indiana experiments} and~\Cref{app: missing experiments}, respectively) capture the performance of the Food drop program and possible future expansions of it. The other two scenarios (in~\Cref{subsec: california experiments} and~\Cref{app: missing experiments}, respectively) are chosen from the West and the South to increase the diversity of the presented results.



\subsection{Indiana experiments}\label{subsec: indiana experiments}

\begin{table}[h]
    \small
    \centering
    \begin{tabular}{l|c|c|c|c|c}
        \hline
        & \textbf{Metric} & \textbf{Alg. 1} & \textbf{Driver Optimal} & \textbf{Greedy} & \textbf{Greedy with cutoff} \\
        \hline
        \multirow{2}{*}{Fairness} & Max m-Envy & $1.0015$ & 2.04 & $1.0007$ & $1.0012$ \\
        & Mean m-Envy & $1.00025$ & 1.175 & $1.00012$ & $1.00020$ \\
        \hline
        \multirow{2}{*}{Driver Efficiency} & Max rel. distance & 2.93 & 1 & 34.3 & 2.5 \\
        & Mean rel. distance & 1.12 & 1 & 2.34 & 1.11 \\
        \hline
    \end{tabular}
    \caption{Fairness and Efficiency statistics for Indiana.}
    \label{tab:indiana}
\end{table}

In~\Cref{tab:indiana}, we see the performance of all algorithms across all metrics in Indiana, and in~\Cref{Indy_Envy_Time.png} (\Cref{app: missing experiments}) we see how the multiplicative envy evolves over time. 
Regarding~\Cref{ourAlgorithm}, all metrics align with or outperform the theoretical guarantees. First, the maximum (and therefore the mean) multiplicative envy is very small. Second,
the driver that is worse off is asked to drive 2.93 times more compared to their optimal route, which is slightly below our theoretical upper bound of 3. On average, drivers only travel 1.12 times more. That is, on average, drivers increase their travel distance by only 12\% compared to their optimal route. 
The driver optimal algorithm achieves, by definition, optimal driver efficiency. However, its maximum and average multiplicative envy are orders of magnitude worse. On the other extreme, the greedy algorithm asks a driver to travel 34.3 times more than the optimal route, while it only provides a marginal improvement on multiplicative envy compared to~\Cref{ourAlgorithm}. 
Finally, greedy with an \emph{optimal} cutoff of 60 miles (that is, we tried different cut-off values and show here the one that achieves the best fairness-efficiency trade-off) shows a marginal improvement over~\Cref{ourAlgorithm} in both metrics. Of course, the lack of theoretical guarantees raises some concerns regarding the robustness of this superior performance under more adversarial settings. These concerns, as we see in the next section, are valid, as substantiated in other states such as California.


\subsection{California experiments}\label{subsec: california experiments}

\begin{table}[h]
    \small
    \centering
    \begin{tabular}{l|c|c|c|c|c}
        \hline
        & \textbf{Metric} & \textbf{Alg. 1} & \textbf{Driver Optimal} & \textbf{Greedy} & \textbf{Greedy with cutoff} \\
        \hline
        \multirow{2}{*}{Fairness} & Max m-Envy & $1.0045$ & 3.94 & $1.0035$ & $1.2$ \\
        & Mean m-Envy & $1.00054$ & 1.3 & $1.00041$ & $1.013$ \\
        \hline
        \multirow{2}{*}{Driver Efficiency} & Max rel. distance & 2.92 & 1 & 79.7 & 5.06 \\
        & Mean rel. distance & 1.06 & 1 & 3.55 & 1.098 \\
        \hline
    \end{tabular}
    \caption{Fairness and Efficiency statistics for California.}
    \label{tab:California}
\end{table}

In~\Cref{tab:California} we see that~\Cref{ourAlgorithm} is again near-optimal across metrics. The driver optimal and greedy algorithm are consistent with respect to their corresponding advantages and disadvantages, which are now even more exacerbated (the envy of the driver optimal algorithm is worse than in Indiana, and so is the driver efficiency of greedy). A key observation is the under-performance of the greedy with cutoff across all metrics. In~\Cref{tab:California} we show a single, specific cutoff value, 60 miles.

\begin{figure}[h]
  \centering
  \includegraphics[width=0.4\linewidth]{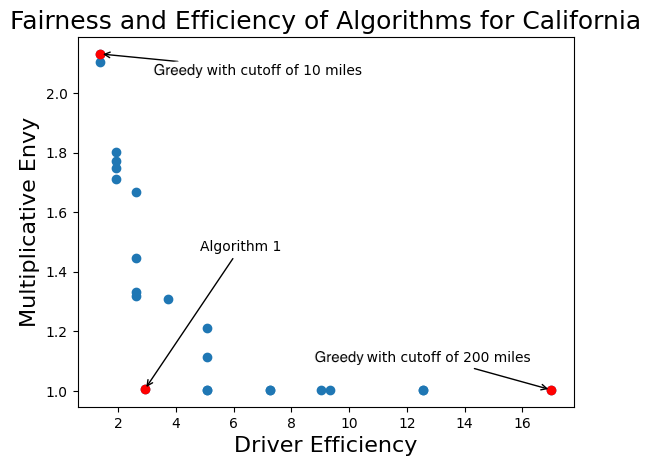} 
  \caption{Fairness-Efficiency trade-offs for the greedy with cutoff family and~\Cref{ourAlgorithm} in California.}
  \label{fig:paretoFrontierOfCutoff}
\end{figure}

\Cref{fig:paretoFrontierOfCutoff} illustrates the trade-off between the maximum multiplicative envy and the maximum/worst driver efficiency for various cutoffs. Clearly,~\Cref{ourAlgorithm}, positioned in the bottom left corner, is not Pareto dominated by any cutoff. See also~\Cref{California_Envy_Time_updated.png} (\Cref{app: missing experiments}) for the behavior of maximum envy over time for the different algorithms.



In~\Cref{fig:CaliforniaMap} we see how food is distributed across California. The leftmost image illustrates the distribution of the food-insecure population across all counties. For the remaining images, green color shades indicate that a county receives food near-proportionately to its food-insecure population; red color shades indicate that a county is under-served, while brown and black color shades indicate that a county is over-served. The driver optimal algorithm has widely different behavior across counties. For the greedy algorithm with a cutoff of 60 miles\footnote{As illustrated in \cref{fig:paretoFrontierOfCutoff}, certain cutoffs enable the greedy-with-cutoff algorithm to perform well in terms of fairness. However, at these cutoffs, the algorithm significantly lacks in truck driver efficiency. In \cref{fig:CaliforniaMap}, we selected a cutoff where greedy-with-cutoff shows comparable truck driver efficiency to Algorithm 1, thus highlighting its shortcomings in the fairness metric.} a particular region in the north-east is under-served.~\Cref{ourAlgorithm} achieves near-perfect coverage across the state.

\begin{figure}[H]
  \centering
  \includegraphics[width=0.95\linewidth]{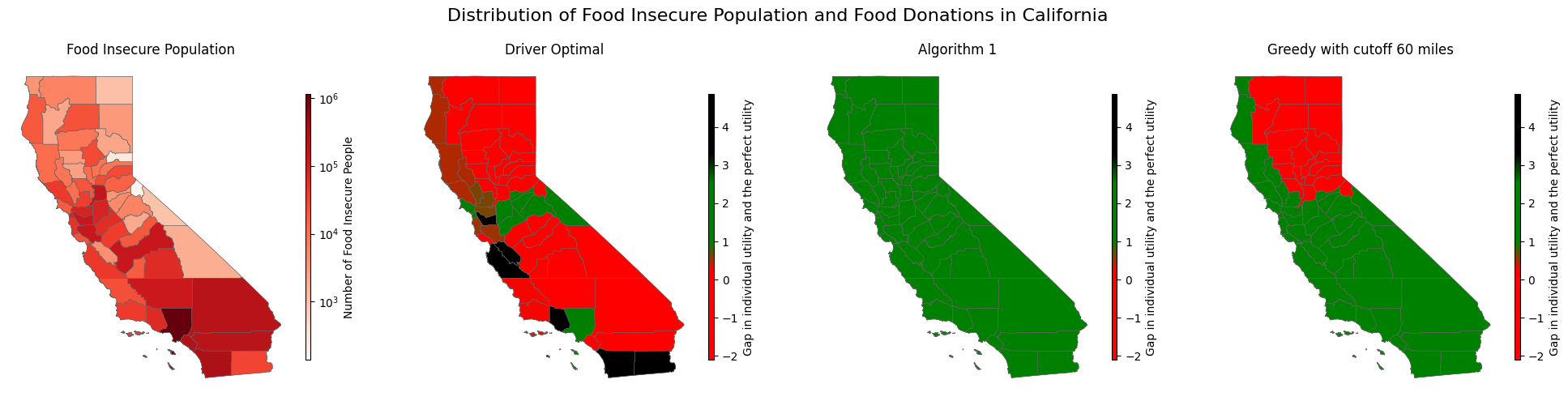} %
  \caption{The first image displays the distribution of the food-insecure population in California. In the subsequent images, the color indicates whether a county receives food near-proportionately (green), more-than-proportionately (brown and black), or under-proportionately (red) to its food-insecure population, under different algorithms.}
  \label{fig:CaliforniaMap}
\end{figure}


\subsection{Testing the $(\alpha, \beta)$-biased hypothesis in practice}
\label{section:alphaBeta}

Recall that in our platform we expect drivers to have an origin and destination at large grocery stores and distribution centers; the number of these is proportional to the total population of a county, which, in turn, correlates strongly with the food insecure population. Let $\calD'$ be the probability distribution that samples a food bank proportionally to the food insecure population it serves, and assume that $\calD'$ is $(\alpha, \beta)$-biased with respect to the actual population of the area each food bank serves. We note that our theoretical results only need bounds on the bias of $\calD'$, not on the bias of the distribution $\calD$ that samples at the county level. For different areas of interest we provide exact calculations of $\alpha$ and $\beta$.

In \cref{tab:alphabeta} we see the values of $\alpha$, $\beta$, as well as $f(\alpha,\beta) = 2 \ln\left( \frac{\alpha \beta -1}{\alpha \beta - \beta} \right)/ \ln\left( \frac{\alpha \beta -1}{\alpha -1} \right)$, the value that controls the error in~\Cref{theorem:gapBound}. As we can see, for all areas, $\alpha$ and $\beta$ are bounded by very small constants. And, in fact, $f(\alpha,\beta)$ is significantly bigger than $1$, which allows for extra slack in the probability distribution from which we assume the origin and destinations are sampled.
However, $\alpha$ and $\beta$ are not as small as $\sqrt{\frac{5}{4}} \approx 1.118$ and $\sqrt{\frac{4}{3}} \approx 1.155$ as~\Cref{theorem:gapBoundWeightedBalls} needs. Nevertheless, our experiments show that \cref{ourAlgorithm} performs exceptionally well in these areas, thus suggesting that further improvements in the tightness of our theoretical guarantees are possible.

\begin{table}[H]
    \centering
    \begin{tabular}{lccccc}
        \hline
        & \textbf{Metric} & \textbf{Indiana} & \textbf{Indiana++} & \textbf{California} & \textbf{Virginia} \\
        \hline
        \multirow{3}{*} & $\alpha$ & $1.12$ & $1.38$ & 1.43 & $1.32$ \\
        & $\beta$ & $1.26$ & 1.39 & $1.35$ & $1.44$ \\
        & $f(\alpha,\beta)$ & $1.62$ & 1.26 & $1.22$ & $1.29$ \\
        \hline
    \end{tabular}
    \caption{Comparison of $(\alpha, \beta)$ values across areas.}
    \label{tab:alphabeta}
\end{table}




%% file: appendix.tex
\section{Missing proofs}\label{app: missing proofs}

\begin{proof}[Proof of~\Cref{corollary:Phi}]
    \begin{align*}
        \EX{}{\Phi(x^s(t+1))- \Phi(x^s(t)) | x^s(t)} &\le \sum_{i=1}^n 2ap_i e^{ax^s_i(t)} \\
        &\le \beta^2 \sum_{i=1}^n 2ap^U_i e^{ax^s_i(t)} \\
        &\le \beta^2 \frac{2a}{n}\Phi(x^s(t)).
    \end{align*}
    The first inequality is implied by~\Cref{lemma: first phi bound} and the fact that for sufficiently small $a$ we have that $Sa \le 1$. The second inequality follows from~\Cref{lemma:abBound}. The last inequality is due to the fact that $p^U_i$'s are increasing with $i$ and $x_i$'s are decreasing.\footnote{Therefore, replacing $p^U_i$'s with $\frac{1}{n}$ gives an upper bound.}
\end{proof}

\begin{proof}[Proof of~\Cref{lemma:Phi_1}]
    First, we upper bound $\sum_{i =1}^N \left( p_i(a + Sa^2)\right)e^{ax^s_i(t)}$:
    \begin{align*}
        \sum_{i =1}^N \left( p_i\left(a + Sa^2\right)\right)e^{ax^s_i(t)} &\le \sum_{i < \frac{3N}{4}} \left( p_i\left(a + Sa^2\right)\right)e^{ax^s_i(t)} + \sum_{i \ge \frac{3N}{4}}^N \left( p_i\left(a + Sa^2\right)\right)e^{0} \\
        &\le \sum_{i < \frac{3N}{4}} \left( p_i\left(a + Sa^2\right)\right)e^{ax^s_i(t)} + 1 \\
        &\le \beta^2 \sum_{i < \frac{3N}{4}} p^U_i\left(a + Sa^2\right)e^{ax^s_i(t)} + 1,
    \end{align*}
    where the first line follows from the condition in the lemma statement ($x^s_{\frac{3N}{4}}(t) \le 0$), the second line follows from picking $a$ sufficiently small ($a + Sa^2 \leq 1$), and the last line follows from~\Cref{lemma:abBound}.
    Where the last inequality holds for a sufficiently small choice of $a$. Let $y_i = e^{ax^s_i(t)}$. 
    
    Finding an upper bound for our expression is equivalent to solving the following LP
    \begin{equation*}
\begin{array}{ll@{}ll}
\text{maximize}  & (a+Sa^2)\beta^2\displaystyle\sum\limits_{i < \frac{3N}{4}} p^U_i \cdot y_i\\
\\
\text{s.t.}& \sum_{i < \frac{3N}{4} }y_i \le \Phi(x^s(t)) \\ 
\\
&  y_{i-1} \ge y_i, \forall 1<i<\frac{3N}{4} 
\end{array}
\end{equation*}

Since $y_i$ is non-increasing and $p^U_i$ is increasing, we have that the maximum is achieved when $y_i = \frac{4\Phi(x^s(t))}{3N}$, $\forall 1 \le i \le \frac{3N}{4} $. In which case the maximum would be $(a+Sa^2)\beta^2\left(\frac{3}{4}\right)^2\frac{4\Phi(x^s(t))}{3N}$. Combining all of the above we have that:

\begin{align*}
    \EX{}{\Phi(x^s(t+1))- \Phi(x^s(t)) | x^s(t)} &\le \left(a+Sa^2 \right)\beta^2\frac{3}{4N}\Phi(x^s(t)) +1 - \left(\frac{a}{N}-S\frac{a^2}{W^2} \right)\Phi(x^s(t)) \\
    &= \left( \left(1+Sa \right)\beta^2\frac{3}{4} - 1+S\frac{a}{N}\right)\frac{a}{N} \Phi(x^s(t)) +1 \\
    &= \left( \beta^2\frac{3}{4}+ \beta^2\frac{3}{4}Sa+ - 1+S\frac{a}{N}\right)\frac{a}{N} \Phi(x^s(t)) +1 \\
    &= \left( \frac{3}{4}\left(\beta^2 -\frac{4}{3} \right)+ \left(\beta^2\frac{3}{4}+\frac{1}{N} \right)aS\right)\frac{a}{N} \Phi(x^s(t)) +1 \\
    &\le \left( -\frac{3}{4}\mu+ \left(1-\frac{3}{4}\mu+\frac{1}{N} \right)aS\right)\frac{a}{N} \Phi(x^s(t)) +1\\
    &\le -\frac{\mu a}{2N} \Phi(x^s(t)) +1,
\end{align*}
where the last inequality follows by picking $a$ small enough.
\end{proof} 

\begin{proof}[Proof of~\Cref{lemma:Psi_1}]
    First, we upper bound $\sum_{i =1}^N \left( p_i(-a + Sa^2)\right)e^{-ax^s_i(t)}$:
    \begin{align*}
        \sum_{i =1}^N \left( p_i\left(-a + Sa^2\right)\right)e^{-ax^s_i(t)} &\le \left(-a + Sa^2\right)\sum_{i \ge \frac{N}{4}} p_i e^{-ax^s_i(t)} \\
        &\le \left(-a + Sa^2\right)\frac{1}{\alpha^2}\sum_{i \ge \frac{N}{4}} p^U_i e^{-ax^s_i(t)},
    \end{align*}
    where the last inequality follows from~\Cref{lemma:abBound}. Now, set $z_i = e^{-ax^s_i(t)}$. Finding an upper bound for our expression is equivalent to solving the following LP:

    \begin{equation*}
\begin{array}{ll@{}ll}
\text{minimize}  & \displaystyle\sum\limits_{i \ge \frac{N}{4}} p^U_i \cdot z_i\\
\\
\text{s.t.}& \sum_{i \ge \frac{N}{4} }z_i \ge \Psi(x^s(t)) -\frac{N}{4} \\ 
\\
&  z_{i-1} \le z_i, \forall i>\frac{N}{4} 
\end{array}
\end{equation*}

Since $z_i$s are increasing we have that the minimum is achieved when $z_i = \frac{4\left(\Psi(x^s(t)) -\frac{N}{4}\right)}{3N}$, in which case the minimum of our original expression would be $(-a+Sa^2)\frac{1}{\alpha^2}\left(1-\frac{1}{4^2}\right)\frac{4\left(\Psi(x^s(t)) -\frac{N}{4}\right)}{3N}$. Combining all of the above we have that:
\begin{align*}
    \EX{}{\Psi(x^s(t+1))- \Psi(x^s(t)) | x^s(t)} &\le \left(-a+Sa^2 \right)\frac{1}{\alpha^2}\frac{15}{16}\left(\frac{4\left(\Psi(x^s(t)) -\frac{N}{4}\right)}{3N} \right) + \left(\frac{a}{N}+S\frac{a^2}{N^2}\right)\Psi(x^s(t)) \\
    &= \left( \left(-1+Sa \right)\frac{1}{\alpha^2}\frac{15}{16}\left(\frac{4}{3}\right) + \left(1+S\frac{a}{N}\right)\right)\frac{a}{N}\Psi(x^s(t)) + \left(a+Sa^2 \right)\frac{1}{\alpha^2}\frac{5}{16} \\
    &\le \left( -\frac{5}{4}\frac{1}{\alpha^2} + \frac{5}{4}\frac{1}{\alpha^2}Sa + 1 +S\frac{a}{N} \right)\frac{a}{N} \Psi(x^s(t)) + 1 \\
    &\le \left(-\frac{5}{4}\mu + \mu\frac{1}{2}   \right)\frac{a}{N} \Psi(x^s(t)) +1 \\
    &\le -\frac{3}{4}\mu \frac{a}{N} \Psi(x^s(t)) +1. \qedhere
\end{align*}
\end{proof}

\begin{proof}[Proof of~\Cref{lemma:Phi_2}]

\begin{align*}
    &\EX{}{\Phi(x^s(t+1))- \Phi(x^s(t)) | x^s(t)} \le \sum_{i =1}^N \left( p_i\left(a + Sa^2\right)-\left(\frac{a}{N} - S\frac{a^2}{N^2}\right) \right) e^{a x^s_i(t)} \\
    &\le \sum_{i \le N/3} \left( p_i\left(a + Sa^2\right)-\left(\frac{a}{N} - S\frac{a^2}{N^2}\right) \right) e^{a x^s_i(t)} + \left(a + Sa^2\right)\sum_{i>N/3} p_i e^{ax^s_i(t)} \\
    &\le \left(\left(a + Sa^2\right)\beta^2\frac{2}{3N} - \left(\frac{a}{N} - S\frac{a^2}{N^2}\right) \right)\Phi_{\le N/3}(x^s(t)) + \left(a + Sa^2\right)\frac{3}{2N}\beta^2 \Phi_{>N/3}(x^s(t)) \\
    &= \left(\frac{2}{3}\left(\beta^2 - \frac{3}{2}\right) +\left(\beta^2\cdot \frac{2}{3}+\frac{1}{N}\right)Sa \right)\frac{a}{N}\Phi_{\le N/3}(x^s(t)) + \left( \frac{3}{2}\beta^2 + \frac{3}{2}\beta^2Sa \right)\frac{a}{N}\Phi_{>N/3}(x^s(t))\\
    &=  -\left(\frac{2}{3}\left(\frac{3}{2}-\beta^2 \right) -\left(\beta^2\cdot \frac{2}{3}+\frac{1}{N}\right)Sa \right)\frac{a}{N}\Phi(x^s(t)) + \left(\frac{5}{6}\beta^2 +1 +\left(\frac{5}{6}\beta^2+\frac{1}{N}\right)Sa\right)\frac{a}{N}\Phi_{>N/3}(x^s(t))\\
    &\le -\frac{1}{18}\frac{a}{N}\Phi(x^s(t)) + \frac{3a}{N}\Phi_{>N/3}(x^s(t)),
\end{align*}
where the third inequality is due to the fact that for $i \le n/3$ we have that $p_i \le \beta^2 \cdot p_i^U\le \beta^2\frac{2}{3N}$ and for a given decreasing $x^s$, $\sum_{i>N/3} p^u_i e^{ax^s_i(t)}$ is maximized for uniform $p_i$.

Thus given that $\EX{}{\Phi(x^s(t+1)) - \Phi(x^s(t))| x^s(t)} \ge -\frac{1}{36} \frac{a}{N}\Phi(x^s(t))$ we have that:

\[\frac{a}{36 N}\Phi(x^s(t)) \le \frac{3a}{N}\Phi_{>N/3}(x^s(t)).\]

Now let $B = \sum_i \max{(0,x^s_i)}$. Then $\Phi_{\ge N/3}(x^s(t)) \le \frac{2N}{3}e^{\frac{3aB}{N}}$ (simply recall that $x^s_i$ are decreasing). Thus:
\[\frac{a}{36 N}\Phi(x^s(t)) \le 2ae^{\frac{3aB}{N}}.\]
On the other hand $x_{\frac{3N}{4}} > 0$ implies that $\Psi(x^s(t)) \ge \frac{N}{4} e^{\frac{4aB}{N}}$. Now if $\Phi(x^s(t)) < \frac{3 \cdot \mu}{32} \Psi(x^s(t))$ then we are done. Otherwise:

\begin{align*}
    2ae^{\frac{3aB}{N}} &\ge \frac{a}{36 N}\Phi(x^s(t)) \ge \frac{3 \cdot \epsilon}{32} \frac{a}{36 N}\Psi(x^s(t)) \ge \frac{3 \cdot \mu}{32} \frac{a}{144}e^{\frac{4aB}{N}}.
\end{align*}

This implies that $e^{\frac{aB}{N}} \le \frac{288}{\frac{3 \cdot \epsilon}{32}}$ It follows that:

\[\Gamma \le (1+\frac{3 \cdot \mu}{32}) \Phi(x^s(t)) \le  (1+\frac{3 \cdot \mu}{32}) 72\left(\frac{288}{\frac{3 \cdot \mu}{32}}\right)^3 N. \]
\end{proof}

\begin{proof}[Proof of~\Cref{lemma:Psi_2}]

\begin{align*}
    &\EX{}{\Psi(x^s(t+1))- \Psi(x^s(t)) | x^s(t)} \le \sum_{i =1}^N \left( p_i\left(-a + Sa^2\right)+\left(\frac{a}{N} + S\frac{a^2}{N^2}\right) \right) e^{-a x^s_i(t)} \\
    &\le \sum_{i > \frac{2N}{3}} \left( p_i\left(-a + Sa^2\right)+\left(\frac{a}{N} + S\frac{a^2}{ N^2}\right) \right) e^{-a x^s_i(t)} + \left(\frac{a}{N} + S\frac{a^2}{ N^2}\right)\sum_{i \le \frac{2N}{3}} e^{-ax^s_i(t)} \\
    &\le \left(\left(-a + Sa^2\right)\left(\frac{4}{3N} - \frac{1}{N^2} \right)\frac{1}{\alpha^2} + \left(\frac{a}{N} + S\frac{a^2}{N^2}\right) \right)\Psi_{> \frac{2N}{3}}(x^s(t)) + \left(\frac{a}{N} + S\frac{a^2}{ N^2}\right) \Psi_{\le \frac{2N}{3}}(x^s(t))\\
    &\le -\frac{a}{30N}\Psi_{> \frac{2N}{3}}(x^s(t)) + \frac{5a}{4N}\Psi_{\le \frac{2N}{3}}(x^s(t))\\
    &\le  -\frac{a}{30N}\Psi(x^s(t)) + \frac{3a}{2N}\Psi_{\le \frac{2N}{3}}(x^s(t)).
\end{align*}

Thus given that $\EX{}{\Psi(x^s(t+1)) - \Psi(x^s(t))| x^s(t)} \ge - \frac{a}{60N}\Psi(x^s(t))$ we have that:

\[\frac{a}{60N}\Psi(x^s(t)) \le  \frac{3a}{2N}\Psi_{\le \frac{2N}{3}}(x^s(t)).\]

Now let $B = \sum_i \max{(0,x^s_i(t))}$. Then $\Psi_{\le \frac{2N}{3}}(x^s(t)) \le \frac{2N}{3}e^{\frac{3aB}{N}}(x^s(t))$ (simply recall that $x^s_i$ are decreasing). Thus:
\[\frac{a}{60N}\Psi(x^s(t)) \le a  e^{\frac{3aB}{N}}.\]
On the other hand $x_{\frac{N}{4}} < 0$ implies that $\Phi(x^s(t)) \ge \frac{N}{4} e^{\frac{4aB}{N}}$. Now, if $\Psi(x^s(t)) < \frac{\mu}{6} \Phi(x^s(t))$ then we are done. Otherwise:

\[a  e^{\frac{3aB}{N}} \ge \frac{a}{60N}\Psi(x^s(t)) \ge \frac{a}{60N} \frac{\mu}{6} \Phi(x^s(t)) \ge \frac{a}{240} \frac{\mu}{6} e^{\frac{4aB}{n}}.\]

This implies that $e^{\frac{aB}{N}} \le \frac{240}{\frac{\mu}{6}}$, and therefore:

\[\Gamma \le (1+\frac{\mu}{6}) \Psi(x^s(t)) \le  (1+\frac{\mu}{6}) 60 \left(\frac{240}{\frac{\mu}{6}}\right)^3 N.\]
    
\end{proof}

\begin{proof}[Missing cases from the proof of~\Cref{lemma: first bound on gamma}]

Here, we consider the missing cases from the case analysis of the proof of~\Cref{lemma: first bound on gamma}.

\textbf{Case 2.1: $\Phi(x^s(t)) < \frac{3 \cdot \mu}{32} \Psi(x^s(t))$.}  In this case, from~\cref{corollary:Phi} and~\cref{lemma:Psi_1} we get:
\begin{align*}
    \EX{}{\Delta \Gamma|x^s(t)} &\le \EX{}{\Delta \Phi|x^s(t)} + \EX{}{\Delta \Psi|x^s(t)}\\
    &\le \frac{8a}{3N}\Phi(x^s(t)) - \frac{3\mu a}{4N} \Psi(x^s(t)) + 1\\ 
    &\le \left(\frac{1}{4}\mu -\frac{3}{4}\mu \right)\frac{a}{N}\Psi(x^s(t)) + 1  \\
    &\le -\frac{a}{2N}\Gamma(x^s(t)) + 1.
\end{align*}
\textbf{Case 2.2: $\Gamma(x^s(t)) < cN$.} In this case, we can combine~\cref{corollary:Phi} and~\cref{corollary:Psi} and get $\EX{}{\Delta \Gamma|x^s(t)} \le \frac{8a}{3N}\Gamma(x^s(t)) \le 6ac \le c -(1-6a)c \le c - \frac{(1-6a)}{N}\Gamma(x^s(t))$.

\textbf{Case 3: $0 \geq x^s_{N/4}(t) \ge x^s_{\frac{3N}{4}}(t)$.} If $\EX{}{\Delta \Psi|x^s(t)} \le -\frac{a}{60N} \Phi(x^s(t))$,~\cref{lemma:Phi_1} implies $\EX{}{\Delta \Gamma|x^s(t)} \le -\min{\{\frac{1}{60}, \frac{\mu}{2}\}}\frac{a}{N}\Gamma(x^s(t))+1$. Otherwise, by~\Cref{lemma:Psi_2}, we have two distinct cases:

\textbf{Case 3.1: $\Psi(x^s(t)) < \frac{\mu}{6} \Phi(x^s(t))$.} In this case, from \cref{corollary:Psi} and \cref{lemma:Phi_1} we get that:
\begin{align*}
    \EX{}{\Delta \Gamma|x^s(t)} &\le \EX{}{\Delta \Phi|x^s(t)} + \EX{}{\Delta \Psi|x^s(t)}\\
    &\le \frac{3a}{2N}\Psi(x^s(t)) - \frac{\mu a}{2N} \Phi(x^s(t)) + 1\\ 
    &\le \left(\frac{3}{2}\frac{\mu}{6} -\frac{\mu}{2} \right)\frac{a}{N}\Phi(x^s(t)) + 1 \\
    &\le -\frac{a}{4N}\Gamma(x^s(t)) + 1.
\end{align*}
\textbf{Case 3.2: $\Gamma(x^s(t)) < c' N$.} In this case we can combine \cref{corollary:Phi} and \cref{corollary:Psi} and get $\EX{}{\Delta \Gamma|x^s(t)}  \le c' - \frac{(1-6a)}{N}\Gamma(x^s(t))$.
\end{proof}

\section{Missing experiments and figures}\label{app: missing experiments}

\subsection{Missing figures}

\begin{figure}[H]
  \centering
  \includegraphics[width=0.7\linewidth]{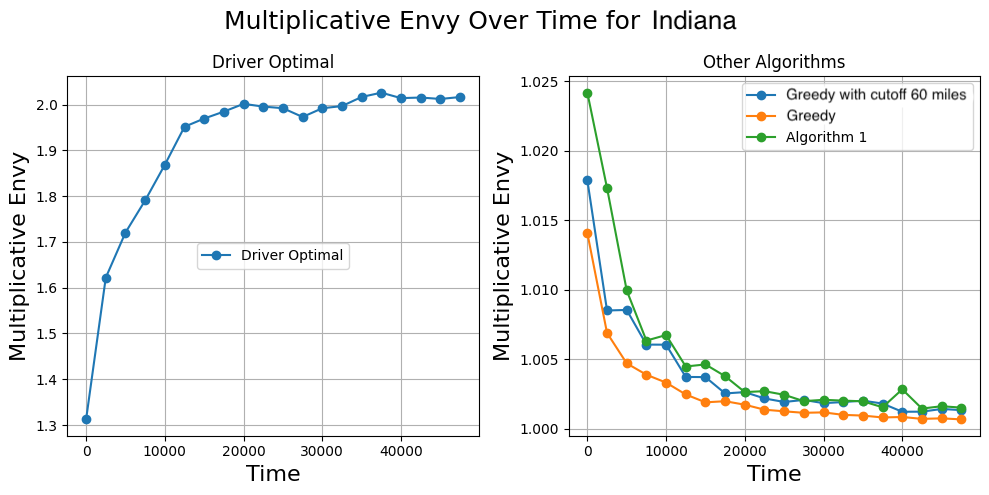} 
  \caption{Maximum multiplicative envy quickly surpasses $1.5$ and converges to $2$ for the driver optimal algorithm, whereas it monotonically and quickly converges to $1$ for the other algorithms.}
  \label{Indy_Envy_Time.png}
\end{figure}

\begin{figure}[H]
  \centering
  \includegraphics[width=0.75\linewidth]{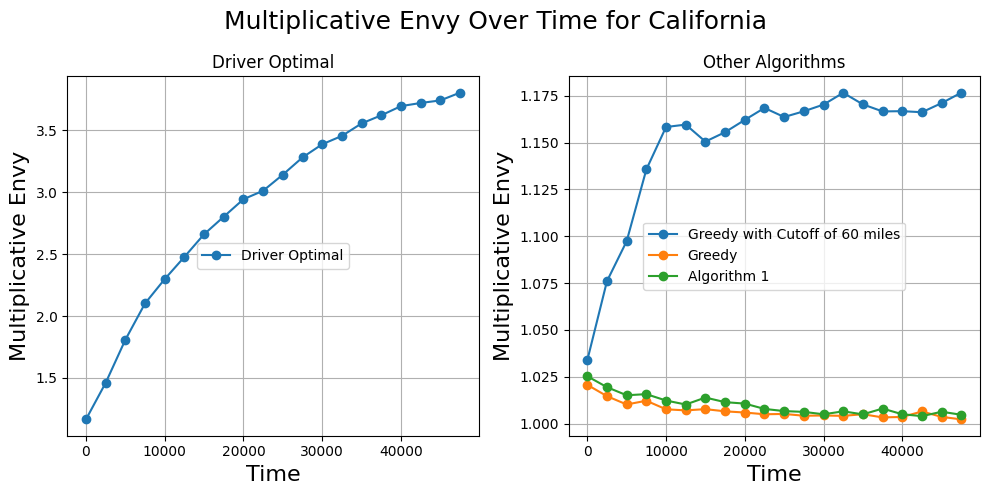} 
  \caption{Maximum envy increases with time for the driver optimal, and greedy with cutoff algorithm, whereas it doesn't grow with time for~\Cref{ourAlgorithm} and greedy.}
  \label{California_Envy_Time_updated.png}
\end{figure} 

\subsection{Indiana++ experiments}\label{subsec: indiana plus experiments}

\begin{table}[H]
    \centering
    \begin{tabular}{l|c|c|c|c|c}
        \hline
        & \textbf{Metric} & \textbf{Alg. 1} & \textbf{Driver Opt.} & \textbf{Greedy} & \textbf{Greedy with cutoff} \\
        \hline
        \multirow{2}{*}{Fairness} & Max m-Envy & $1.004$ & 9.93 & $1.0017$ & $1.0037$ \\
        & Mean m-Envy & $1.00052$ & 1.57 & $1.00024$ & $1.00048$ \\
        \hline
        \multirow{2}{*}{Driver Efficiency} & Max rel. distance & 2.93 & 1 & 130.46 & 3.38 \\
        & Mean rel. distance & 1.1 & 1 & 2.8 & 1.08 \\
        \hline
    \end{tabular}
    \caption{Fairness and Efficiency comparison for Indiana++ (Indiana, Kentucky, and Illinois combined).}
    \label{tab:midwest}
\end{table}

In \cref{tab:midwest} we see that our results in ``Indiana++'' are similar to the ones presented for the state of Indiana; for ``Greedy with Cutoff,'' we show the metrics for a cutoff of 60 miles. We can see that the guarantees of the benchmarks are worse, indicating that they do not scale well as we increase the distances, the number of nodes in the underlying graph, and the number of food banks that need to be served.

\begin{figure}[H]
  \centering
  \includegraphics[width=0.75\linewidth]{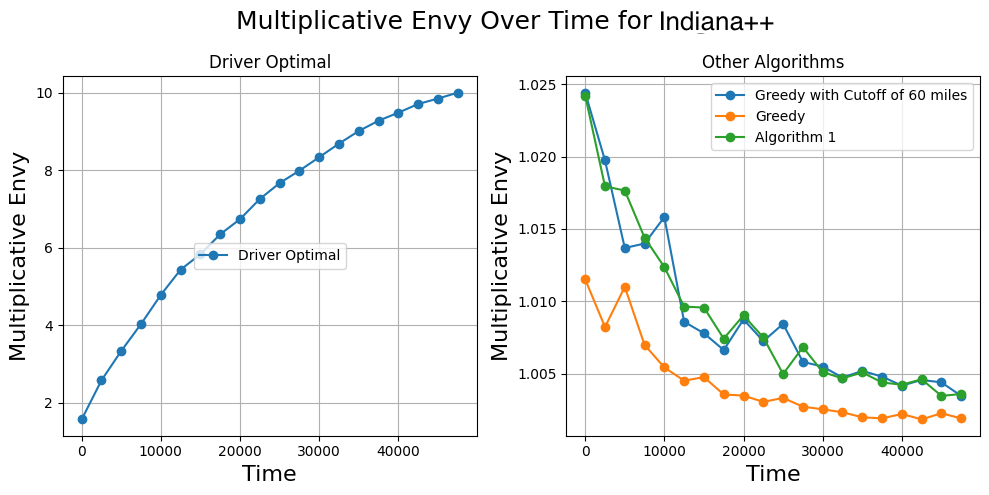} 
  \caption{Maximum envy increases over time for the driver optimal algorithm, whereas it doesn't grow with time for~\Cref{ourAlgorithm}, Greedy, and Greedy with a cutoff.}
  \label{Indy_Envy_Time_updated.png}
\end{figure}

\subsection{Virginia experiments}\label{subsec: virginia experiments}
Table for statistics of fairness and driver efficiency. 

\begin{table}[H]
    \centering
    \begin{tabular}{l|c|c|c|c|c}
        \hline
        & \textbf{Metric} & \textbf{Alg. 1} & \textbf{Driver Opt.} & \textbf{Greedy} & \textbf{Greedy with cutoff} \\
        \hline
        \multirow{2}{*}{Fairness} & Max m-Envy & $1.00079$ & 3.86 & $1.00035$ & $1.1$ \\
        & Mean m-Envy & $1.00017$ & 1.43 & $1.000076$ & $1.017$ \\
        \hline
        \multirow{2}{*}{Driver Efficiency} & Max rel. distance & 2.92 & 1 & 190.2 & 13.99 \\
        & Mean rel. distance & 1.09 & 1 & 2.5 & 1.048 \\
        \hline
    \end{tabular}
    \caption{Fairness and Efficiency comparison for Virginia.}
    \label{tab:example}
\end{table}

In Virginia, we see a similar pattern as the one in California.~\Cref{ourAlgorithm} is not dominated by any other algorithm on the desired metrics. For ``Greedy with Cutoff,'' we show the metrics for a cutoff of 60 miles. The performance for different cutoffs can be seen in~\Cref{fig: Virginia}.

\begin{figure}[H]
  \centering
  \includegraphics[width=0.5\linewidth]{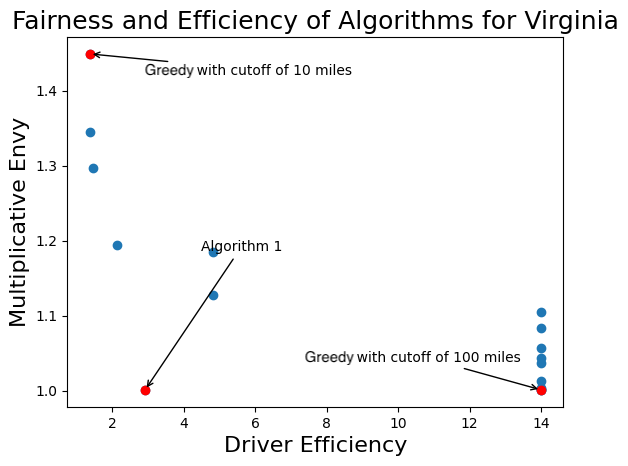} 
  \caption{Comparing~\Cref{ourAlgorithm} to Greedy with various cutoffs for Virginia; no cutoff Pareto dominates~\Cref{ourAlgorithm}.}
  \label{fig: Virginia}
\end{figure}